\documentclass[letterpaper, 10 pt, conference]{ieeeconf}  %
\IEEEoverridecommandlockouts                              % This command is only needed if 
\usepackage{graphicx}
\usepackage{amsmath}
\usepackage{amssymb}
\usepackage{mathtools}

\usepackage{amsfonts}
\usepackage{caption}
\usepackage{subcaption}
\usepackage{color}
\usepackage[ruled,vlined]{algorithm2e}
\newtheorem{problem}{Problem}
\newtheorem{definition}{Definition}

\newtheorem{theorem}{Theorem}
\newtheorem{lemma}{Lemma}
\newtheorem{remark}{Remark}

\title{\LARGE \bf Field Estimation using Robotic Swarms through Bayesian Regression and Mean-Field Feedback}

\author{Tongjia Zheng and Hai Lin%
\thanks{*This work was supported by the National Science Foundation under Grant No. IIS-1724070, CNS-1830335, IIS-2007949.}
\thanks{Tongia Zheng and Hai Lin are with the Department of Electrical Engineering, University of Notre Dame, Notre Dame, IN 46556, USA (e-mail: tzheng1@nd.edu, hlin1@nd.edu.).}
}

\begin{document}

\maketitle

\thispagestyle{empty}
\pagestyle{empty}

\begin{abstract}
Recent years have seen an increased interest in using mean-field density based modelling and control strategy for deploying robotic swarms. 
In this paper, we study how to dynamically deploy the robots subject to their physical constraints to efficiently measure and reconstruct certain unknown spatial field (e.g. the air pollution index over a city). 
Specifically, the evolution of the robots' density is modelled by mean-field partial differential equations (PDEs) which are uniquely determined by the robots' individual dynamics.
Bayesian regression models are used to obtain predictions and return a variance function that represents the confidence of the prediction.
We formulate a PDE constrained optimization problem based on this variance function to dynamically generate a reference density signal which guides the robots to uncertain areas to collect new data, and design mean-field feedback-based control laws such that the robots' density converges to this reference signal.
We also show that the proposed feedback law is robust to density estimation errors in the sense of input-to-state stability.
Simulations are included to verify the effectiveness of the algorithms.
\end{abstract}

%%%%%%%%%%%%%%%%%%%%%%%%%%%%%%%%%%%%%%%%%%%%%%%%%%%%%%%%%%%%%%%%%%%%%%%%%%%%%%%%
\section{Introduction}
In this paper, we study the deployment and control problem of robotic swarms to quickly and efficiently measure and reconstruct an underlying spatial field. 
Our work is motivated by environment monitoring problems, in which one needs to dynamically deploy mobile sensors to collect data and reconstruct certain quantity (e.g. the air pollution index).

%The challenge lies in that on one hand, we want the robots to take measurements in uncertain areas in order to obtain a more confident prediction of the unknown quantity, while on the other hand, we need to consider the inherent dynamics of the robots and design motion commands in a more efficient way. 
%This work studies how to integrate machine learning techniques into the recent PDE-based control scheme for robotic swarms in order to dynamically deploy the mobile sensors based on the real-time monitoring performance. 
% Considering the real-time requirements of monitoring tasks, we study how to dynamically generate reference signals and deploy the robots according to the real-time monitoring performance in order to take measurements and reconstruct the unknown quantity in a more efficient way.

Employing a large group of mobile sensors provides superior robustness and efficiency, but also poses significant challenges from a control theory perspective. 
Many methods have been proposed in the literature for controlling robotic swarms, such as graph theory based approaches \cite{olfati2007consensus}, game theory (especially mean-field game theory) based design \cite{lasry2007mean}, distributed optimal control based motion planning \cite{foderaro2016distributed}. 
Our work is inspired by the mean-field based modelling and control strategy for robotic swarms \cite{elamvazhuthi2019mean}. 
Earlier efforts need to partition the spatial domain and use an abstracted Markov chain model \cite{accikmecse2012markov}, which suffers from the state explosion issues.   
More recent work that follows the ``mean-field'' philosophy but provides more powerful modelling capabilities is by using partial differential equations (PDEs) \cite{eren2017velocity, krishnan2018distributed, elamvazhuthi2018pde, zheng2020complex}. 
%recent years have seen a growing trend of using compact mean-field models for describing their spatial configuration and using (probability) density-based control strategies for deploying the robots \cite{elamvazhuthi2019mean}.
%One representative makes use of Markov chains, which partitions the spatial domain into disjoint cells over which a probability distribution is defined. One then uses convex optimization \cite{accikmecse2015markov} or density feedback laws \cite{bandyopadhyay2017probabilistic} to compute the transition rates between the cells in order to stabilize a target probability distribution. A well-known limitation is that due to the spatial abstraction, inherent robotic dynamics are largely ignored in this approach. 
%Another approach that follows the ``mean-field'' philosophy but provides more powerful modelling capabilities is by using PDEs \cite{hamann2008framework, meurer2011finite, qi2014multi}.
In this approach, individual robots are modelled by stochastic differential equations and their spatial distribution satisfies an associated PDE, such as the Fokker-Planck equation.
%Although higher-dimensional PDEs are difficult to control, the technique of density feedback proves particularly promising because of its stability guarantee, and has been extensively pursued in recent years \cite{eren2017velocity, krishnan2018optimal, elamvazhuthi2018bilinear}. Hence, the PDE models and density feedback control strategy will be the core of this work.
Existing mean-field PDE approaches usually assume that a target distribution is given and the goal is to design proper mean-field feedback laws to guide the swarm towards the target distribution. 
However, in a field estimation task, the target distribution is not known as it is related to the underneath function/quantity to be measured.

%\a more desirable property is that the robots can automatically deploy themselves based on the real-time monitoring performance.

To estimate an unknown field, the robots need to move to uncertain areas to collect more data.
This goal is related to the resampling problem in regression, which is usually fulfilled using Bayesian regression (BR) \cite{gelman2013bayesian}.
% By defining a prior distribution over certain family of models, these algorithms generate a distribution for its prediction, which can be used to determine the locations for resampling.
For example, resampling can take place at locations with large predictive variance.
However, such resampling techniques are not directly suitable for field estimation tasks because of the inherent dynamics and physcial constraints of the robots -- we cannot instantly change the sampling positions. 
Instead we have to design control laws to physically drive them to the new positions. 
Meanwhile, we need to do it in an energy-efficient and robust manner. 
%Hence, we need to address the problem of how to consider the robotic dynamics in the resampling process.

This motivates us to integrate BR models into the PDE-based control design to achieve automatic deployment based on the real-time prediction quality.
Specifically, BR models are used to construct a prediction for the unknown function, as well as a variance function that represents its predictive confidence. 
We formulate a PDE-constrained optimization problem based on this variance function to generate a reference model which encodes the desired evolution of the robots' density in order to improve the prediction quality.
Then, we design mean-field feedback laws for the robots such that their density evolves according to the reference model, and analyze its robustness in terms of density estimation errors using the notion of input-to-state stability.

The problem of field estimation and deployment has also been studied in recent works \cite{de2018optimal, krishnan2018distributed, elamvazhuthi2018pde, morelli2019integrated}.
In \cite{elamvazhuthi2018pde, morelli2019integrated}, the authors also consider a field estimation problem followed by an optimal control formulation for swarm deployment.
However, the numerical solutions are usually open-loop and have robustness issues due to environmental uncertainty.
In \cite{de2018optimal, krishnan2018distributed}, the authors formulate an optimal mass transport problem followed by mean-field feedback strategies.
Our work differs in that we not only design closed-loop control laws but also study their robustness property.
In summary, our contribution includes two aspects. 
First, we present a candidate framework for integrating machine learning techniques into the mean-field PDE approach for automatically deploying the robots based on real-time performance.
Second, we propose mean-field feedback laws for density tracking problems and prove their robustness to density estimation errors.
In this way, we can formally guarantee that the real-time deployment requirement generated by the machine learning algorithms will be fulfilled.

The rest of the paper is organized as follow. Section \ref{section:preliminaries} introduces some preliminaries. Section \ref{section:problem formulation} gives the problem formulation. Section \ref{section:main results} is our main results in which we use BR models to generate reference models and design mean-field feedback laws for the PDE system to converge to the reference model. Section \ref{section:simulation} presents an agent-based simulation to verify the effectiveness.

\section{Preliminaries}\label{section:preliminaries}
\subsection{Notation}
Let $\mathbb{R}^n$ be the $n$-dimensional Euclidean space. We denote by $\Omega$ an open, bounded, and connected subset of $\mathbb{R}^n$, with boundary $\partial\Omega$. Let $\mu$ be the Lebesgue measure. For $1 \leq p<\infty$, denote by $L^{p}(\Omega)$ the space of functions $f$ such that $\int_{\Omega}|f(x)|^{p} d \mu<\infty$, endowed with the norm $\|f\|_{L^{p}(\Omega)}=\left(\int_{\Omega}|f(x)|^{p} d \mu\right)^{1 / p}$. The gradient and Laplacian of a scalar function $f(x)$ are denoted by $\nabla f$ and $\Delta f$, respectively. The divergence of a vector field $\mathbf{F}$ is denoted by $\nabla \cdot \mathbf {F}$. For $1 \leq p<\infty$ and $k \in \mathbb{N}$, denote by $W^{k, p}(\Omega)$ the Sobolev space of functions $f \in L^{p}(\Omega)$ having weak derivatives $D^\alpha f$ in $L^{p}(\Omega)$ for all multi-indices $\alpha$ of length $|\alpha| \leq k$, endowed with the norm $\|f\|_{W^{k, p}(\Omega)}=\left(\sum_{|\alpha| \leq k} \int_{\Omega}\left|\nabla^{\alpha} f(x)\right|^{p} d \mu\right)^{1 / p}$.

\subsection{Input-to-state stability}
Input-to-state stability (ISS) is a stability notion for studying nonlinear systems with external inputs \cite{sontag1995characterizations}. 
To define the ISS concept we need to introduce the following classes of comparison functions \cite{dashkovskiy2013input}:
\begin{align*}
    \mathcal{K} &:= \{\gamma :\mathbb{R}_{+} \to \mathbb{R}_{+}|\gamma \text{ is continuous and strictly} \\
    &\quad\quad \text{increasing}, \gamma (0)=0, \text{and } \gamma (r)>0 \text{ for }r>0\}\\
    \mathcal{K}_{\infty} &:=\{\gamma \in \mathcal{K}|\gamma \text{ is unbounded}\}  \\
    \mathcal{L} &:= \{\gamma  : \mathbb{R}_{+} \to \mathbb{R}_{+} | \gamma \text{ is continuous and strictly } \\
    &\quad\quad \text{decreasing with } \lim_{t\to \infty} \gamma (t) = 0\}  \\
    \mathcal{KL} &:=\{\beta:\mathbb{R}_{+} \times\mathbb{R}_{+} \to\mathbb{R}_{+} |\beta \text{ is continuous}, \beta(\cdot,t)\in \mathcal{K}, \\
    &\quad\quad \beta(r,\cdot)\in \mathcal{L}, \forall t\geq 0, \forall r>0\}.
\end{align*}

\begin{definition}
\cite{dashkovskiy2013input} Consider a control system $\Sigma = (X, U, \phi)$ consisting of normed linear spaces $(X, \|\cdot\|_X)$ and $(U, \|\cdot\|_U)$, called the state space and the input space, endowed with the norms $\|\cdot\|_X$ and $\|\cdot\|_U$ respectively, and a transition map $\phi:\mathbb{R}_{+}\times X \times U \to X$. The system is called ISS if there exist functions $\beta\in \mathcal{KL}$ and $\gamma\in \mathcal{K}$, such that
$$
\|\phi(t,\phi_0,u)\|_X \leq \beta(\|\phi_0\|_X,t) + \gamma(\|u\|_U),
$$
holds $\forall \phi_0\in X$, $\forall t\geq 0$ and $\forall u\in U$. It is called locally input-to-state stable (LISS), if there also exists constants $\rho_{x}, \rho_{u}>0$ such that the above inequality holds $\forall \phi_{0} :\left\|\phi_{0}\right\|_{X} \leq \rho_{x}, \forall t \geq 0$ and $\forall u \in U :\|u\|_{U} \leq \rho_{u}$.
\end{definition}

The following lemma provides a tool for verifying the (L)ISS property by constructing Lyapunov functionals.

\begin{lemma} \label{lmm:ISS Lyapunov functional}
\cite{dashkovskiy2013input} If there exists a continuous functional $V : X \to \mathbb{R}_{+}$, functions $a_1,a_2,a\in  \mathcal{K}_\infty$, $\rho \in  \mathcal{K}$, and constants $\rho_{x}, \rho_{u}>0$, such that:
$$
a_1(\|x\|_X) \leq V (x) \leq a_2(\|x\|_X), \forall x \in X,
$$
and $\forall \phi_{0} :\left\|\phi_{0}\right\|_{X} \leq \rho_{x}$, $\forall u \in U :\|u\|_{U} \leq \rho_{u}$, $\Dot{V}(x)$ satisfies
$$\Dot{V}(x) \leq -a(\|x\|_X), \forall \|x\|_X \geq \rho(\|u\|_U),$$
then the system is LISS. If $\rho_{x}=\infty$ and $\rho_{u}=\infty$, then the system is ISS. The corresponding functional $V$ is called an (L)ISS Lyapunov functional. 
\end{lemma}

\section{Problem Formulation}\label{section:problem formulation}
This work studies the problem of dynamically deploying a robotic swarm to measure and construct an unknown function $f(x)$ on a domain $\Omega\subset\mathbb{R}^m$. 
% To have a concrete understanding, one can think of $\Omega$ as a city and $f(x)$ as the air pollution index. 
% A group of robots (such as UAVs) are sent to the city to measure the air pollution index. 
Denote by $\{X_i(t)\}_{i=1}^n\subset\mathbb{R}^m$ the robots' positions, where $n$ is the population. 
Each robot obtains a noisy measurement $y_i$ of $f(X_i)$ such that $y_i=f(X_i)+\epsilon$, where $\epsilon$ follows an i.i.d. Gaussian distribution. 
We assume a sampling period $\Delta t$ so that measurements are taken at discrete time $k$ with $t=k\Delta t$. 
At each step $k$, we obtain a data set $D_k=\left\{\left(X_{i}(k),y_{i}(k)\right)|i=1,\dots,n\right\}$. 
We denote $\mathcal{D}_k=\cup_{j=1}^kD_j$ to represent all the data collected by time $k \Delta t$.
The robots' motion is assumed to satisfy
\begin{equation} \label{eq:Langevin equation}
    dX_i=v(X_i,t)dt+\sqrt{2\sigma(X_i,t)}dB_t, \quad i = 1,\dots,n,
\end{equation}
where $X_i\in\Omega$ is the position of the $i$-th robot, $v\in\mathbb{R}^m$ is the velocity field that acts on the robot, $B_t\in\mathbb{R}^m$ is an $m$-dimensional Wiener process which represents stochastic motions, and $\sqrt{2\sigma}\in\mathbb{R}$ is the standard deviation.

The density of the robots, denoted by $p(x,t)$, is known to be governed by the following Fokker-Planck equation:
\begin{align} \label{eq:FP equation}
\begin{split}
     \partial_t p =-\nabla\cdot(vp) + \Delta(\sigma p) &\quad\text{in}\quad \Omega\times(0,\infty), \\
    p=p_0 &\quad\text{on}\quad \Omega\times\{0\},\\
    \boldsymbol{n} \cdot(\nabla(\sigma p)-vp)=0 &\quad\text{on}\quad \partial\Omega\times(0,\infty),
\end{split}
\end{align}
where $\boldsymbol{n}$ is the unit inner normal to the boundary $\partial\Omega$, and $p_0$ is the initial density. The last equation is a \textit{reflecting boundary condition} to confine the swarm within $\Omega$. 

\begin{problem}
Our goal to design the velocity field $v$ to guide the robots' movements in order to efficiently reconstruct the unknown $f(x)$ with the increasingly rich data collection $\mathcal{D}_k$.
\end{problem}

This problem is essentially a resampling problem -- we want to take new measurements to construct a better prediction $\Bar{f}$. 
It however poses additional challenges because we need to consider the robots' physical constraints and design suitable motion commands to make the resampling more efficient.

\section{Main results}\label{section:main results}
\subsection{Generation of reference models}\label{section:generate reference model}
This section studies the problem of generating reference models for the swarm, i.e. where to resample.
Intuitively, we expect the robots to move to the areas where the prediction $\Bar{f}$ is less confident.
BR models turn out to fulfill this purpose because besides predicting the function value $\Bar{f}(x)$, they also return the variance $\operatorname{Var}[\Bar{f}(x)]$ to represent its confidence.
Our objective is to minimize the variance by resampling $f$.

Specifically, given $\mathcal{D}_k$, we use BR models (in particular Gaussian process regression models \cite{rasmussen2003gaussian}) to obtain a prediction $\Bar{f}_k(x)$ and its associated variance $\mathcal{V}_k(x):=\operatorname{Var}[\Bar{f}_k(x)]$ for all $x\in\Omega$.
% Considering the motivating example of reconstructing the air pollution index, we choose the SE covariance function \eqref{eq:squared exponential kernel}.
% (Different covariance functions can be used depending on specific problems.)
% It is easy to verify that with this choice, the predictive variance $\mathcal{V}(x)$ is a smooth function of $x$.
We want the robots to move to areas with larger $\mathcal{V}_k(x)$.
We define $\mathcal{W}_k(x)=\max\{\mathcal{V}_k(x)-\eta,0\}$ and construct a target density $p_f(x,k\Delta t) = \frac{\mathcal{W}_k(x)}{\int_\Omega\mathcal{W}_k(x)dx}$, where $\eta$ is a small tolerance to ensure that the algorithm terminates.
We would like to formulate an optimization problem for \eqref{eq:FP equation} to reach $p_f$.
An optimization formulation enables us to impose many practical requirements, such as penalizing the high density area to avoid concentration (which helps reducing robot-robot collision) or penalizing the velocity to save energy.
Note that $p_f$ is time-varying because $\mathcal{D}_k$ grows with time.
However, since the robots move continuously, $\mathcal{V}_k(x)$ does not change significantly within a short interval, for which there is no need to perform regression and update the reference model at every step.
We choose to update the reference model periodically with a period $T$ ($T\gg\Delta t$), and within each period, $p_f$ is held fixed.
Let $t_c$ be the current time and $t_f:=t_c+T$.
We formulate an optimization problem:
\begin{align}\label{eq:optimization problem}
\begin{split}
    &J=\int_{\Omega}\phi\big(p(t_{f}),t_{f}\big)dx+\int_{t_c}^{t_c+T} \int_{\Omega}L(p,v,t)dxdt \\
    &\text{s.t. } \partial_t p=-\nabla \cdot[v(x,t)p(x,t)] \quad\text{in } \Omega \\
    &\qquad p(x,t_c)=p_0 \\
    &\qquad vp\cdot \mathbf{n} = 0 \quad\text{on } \partial\Omega
\end{split}
\end{align}
where the constraint is chosen to be a transport equation.
% For example, we may choose
% $$
% J(p,v,t)=\left.\int_{\Omega} w_{p}(p_f-p)^{2}\right|_{t_{f}}dx + \int_{\Omega}\int_{t_{0}}^{t_{f}} e^{\frac{w_{v}}{2}\left(v_{x}^{2}+v_{y}^{2}\right)}dtdx
% $$
% which forces the terminal state to be close to $p_f$ and also punishes the velocity magnitude.
(Certainly we can subject to the \eqref{eq:FP equation}, but it will unnecessarily complicate the optimality conditions to be derived and its numerical solution.)
% As to be seen in the simulation, the Neumann condition can be easily satisfied using sine series parametrization.
% Subjecting to a transport equation facilitates not only the derivation of the optimality conditions and but also the tracking control design in the next section.
Problem \eqref{eq:optimization problem} is a PDE-constrained optimization problem. 
We follow the procedure in \cite{rudd2013generalized} to obtain the necessary conditions for optimal solutions:
\begin{equation}\label{eq:optimality conditions}
\begin{array}{cl}
    \text{State equation:} &\displaystyle\frac{\partial p}{\partial t} =-\nabla\cdot(vp),\\
    &\text{s.t. } p(x,0)=p_0(x) \\
    &\qquad vp\cdot \mathbf{n} = 0 \quad\text{on } \partial\Omega \\
    \text{Co-state equation:} &\displaystyle\frac{\partial\lambda}{\partial t} = \frac{\partial L}{\partial p}-\nabla\lambda\cdot v,\\
    &\text{s.t. } \lambda(x,t_f)=\frac{\partial\phi}{\partial p}|_{t_f} \\
    &\qquad \lambda = 0 \quad\text{on } \partial\Omega \\
    \text{Optimization condition:} &\displaystyle \frac{\partial L}{\partial v}=p\nabla\lambda,
\end{array}
\end{equation}
{\color{black}
where we corrected a mistake on the boundary condition in \cite{rudd2013generalized}. 
(Note that there is lack of math rigor when applying the calculus of variations for finite dimensional systems to a PDE. 
We will formally study its well-posedness issue in our future work.)
% (The complete derivation can be found in an extended version on arXiv.)
% The optimality conditions \eqref{eq:optimality conditions} don't have closed-form solutions in general and have to be solved numerically \cite{rudd2013generalized}.
% A generalize reduced gradient method is presented in  to solve these equations.
A solution of \eqref{eq:optimality conditions} represents a locally optimal trajectory of $p_\text{r}$ and $v_\text{r}$, so we obtain a reference model:
}
\begin{align}\label{eq:reference model}
\begin{split}
    \partial_tp_\text{r}=-\nabla\cdot(v_\text{r}p_\text{r}) &\quad\text{in}\quad \Omega\times(0,\infty), \\
    p_\text{r}=p_0 &\quad\text{on}\quad \Omega\times\{0\},\\
    \qquad v_\text{r}p_\text{r}\cdot \mathbf{n} = 0 &\quad\text{on}\quad \partial\Omega\times(0,\infty),
\end{split}
\end{align}
which describes the desired density evolution for the robots.
The reference control law $v_\text{r}$ obtained in this way is open-loop.
The remaining task is to design feedback laws $v$ for \eqref{eq:FP equation} such that its solution $p$ converges to the solution of \eqref{eq:reference model}.

% \begin{remark}
% Note that the robots are consecutively collecting new data when they are moving to new areas. Therefore, we should recompute the prediction and its variance, and reconstruct the reference model after certain time period.
% \end{remark}

\subsection{Density tracking control}
In this section, we study how to design $v$ for the robots such that their density $p$ evolves according to the reference model.
Our deign is inspired by the recent work \cite{eren2017velocity}, where a \textit{mean-field feedback}, namely designing $v$ as a function of $p$, was proposed.  
Given desired trajectories of $p_{\text{r}}$ and $v_\text{r}$ from \eqref{eq:reference model},
we define the tracking error as $\Phi=p-p_{\text{r}}$ and design $v$ such that $\Phi$ satisfies the diffusion equation:
\begin{equation} \label{eq:diffusion equation}
    \partial_t\Phi(x,t) = \nabla\cdot[\alpha(x,t)\nabla\Phi(x,t)],
\end{equation}
where $\alpha$ is the diffusion coefficient. 
Under mild conditions on $\alpha$, its solution converges to a constant function, which will be 0 because for any $t$, $\int_{\Omega}\Phi dx=\int_{\Omega} pdx-\int_{\Omega} p_\text{r}dx=1-1=0$.
We propose the mean-field feedback law:
\begin{equation} \label{eq:density feedback law}
    v=-\frac{\alpha(x,t)\nabla(p-p_{\text{r}})-\nabla(\sigma p)-v_{\text{r}}p_{\text{r}}}{p},
\end{equation}
where $\alpha$ is a design parameter that can be used to locally adjust the velocity magnitude.
We require $\sup_{x\in\Omega,t\geq0}\alpha(x,t)<\infty$ and $\inf_{x\in\Omega,t\geq0}\alpha(x,t)>0$.
The convergence property of \eqref{eq:density feedback law} is given below.

\begin{theorem}[Exponential stability]
\label{thm:exponential stability}
Consider the PDE system \eqref{eq:FP equation} with control law \eqref{eq:density feedback law}.
If the solution satisfies $p>0$, then $\|\Phi\|_{L^2(\Omega)}\to0$ exponentially.
% satisfies
% $$\|\Phi(x,t)\|_2\leq \|\Phi(x,0)\|_2 e^{-\beta t},$$
% where $\beta = \frac{\sqrt{2\alpha}}{C}$ and $C$ is a positive constant depending only on $\Omega$.
\end{theorem}
\begin{proof}
Substituting \eqref{eq:density feedback law} into \eqref{eq:FP equation}, we obtain 
\begin{align*} 
\begin{split}
    \partial_t\Phi = \nabla\cdot(\alpha\nabla\Phi) & \quad\text{in}\quad \Omega\times(0,\infty), \\
    \Phi=\Phi_0 & \quad\text{on}\quad \Omega\times\{0\},\\
    \boldsymbol{n} \cdot\nabla\Phi=0 & \quad\text{on}\quad \partial\Omega\times(0,\infty),
\end{split}
\end{align*}
which is a diffusion equation.
Its stability is well-known in the PDE literature.
A proof can be found in \cite{zheng2020transporting}.
% To show its exponential stability, consider a Lyapunov functional
% $$
% V(t)=\frac{1}{2}\|\Phi\|_{L^2(\Omega)}^2=\frac{1}{2}\int_{\Omega}\Phi^2 dx. 
% $$
% Define $\alpha_{\text{min}}(t):=\inf_{x\in\Omega}\alpha(x,t)>0$. 
% We have
% \begin{align*}
%     \Dot{V}(t) &= \int_{\Omega}\Phi \partial_t\Phi dx\\
%     &= \int_{\Omega}\Phi \nabla\cdot[\alpha(x,t)\nabla\Phi]dx\\
%     &= \int_{\partial{\Omega}}\Phi[\alpha(x,t)\nabla\Phi\cdot\boldsymbol{n}]ds-\int_{\Omega}\alpha(x,t)\nabla\Phi\cdot\nabla\Phi dx\\ 
%     &\quad \text{(by the divergence theorem)}\\
%     &= - \alpha_{\text{min}}(t)\int_{\Omega} |\nabla\Phi|^2 dx \quad \text{(by the boundary condition)}\\
%     &\leq -\frac{\alpha_{\text{min}}(t)}{C^2} \int_{\Omega} |\Phi|^2 dx, \quad \text{(by the Poincar\'e inequality)}
% \end{align*}
% where $C$ is the positive constant from the Poincar\'e inequality which only depends on $\Omega$. 
% Since $\alpha$ has a uniform positive lower bound, we obtain exponential stability.
\end{proof}

To be well-defined, the control law \eqref{eq:density feedback law} requires that $p>0$. 
This requirement can be satisfied if we replace $p$ with an estimate $\Hat{p}$. 
In this case, the control law is given by
\begin{equation}\label{eq:density feedback law using estimation}
    v=-\frac{\alpha(x,t)\nabla(\Hat{p}-p_{\text{r}})-\nabla(\sigma\Hat{p})-v_{\text{r}}p_{\text{r}}}{\Hat{p}},
\end{equation}
where $\Hat{p}$ is obtained using kernel density estimation \cite{rao1986nonparametric}, i.e.
{\color{black}
\begin{equation} \label{eq:KDE}
\hat{p}(x,t) = \frac{1}{n h^{m}} \sum_{i=1}^{n} K\left(\frac{1}{h}\left(x-X_{i}(t)\right)\right),
\end{equation}
where $K(x)=\frac{1}{(2 \pi)^{m/2}} \exp \left(-\frac{1}{2}x^\intercal x\right)$ is the Gaussian kernel.
% (Many boundary correction methods exist for refining $\Hat{p}$ to have compact support \cite{silverman1986density}, so we shall not worry about this.)
% As a result, \eqref{eq:FP equation} under \eqref{eq:density feedback law using estimation} is well-posed.
\begin{remark}
We should point out that the proposed framework is essentially centralized because we require a communication center that communicates with the robots to perform regression, solve the optimization problem and estimate the density.
The deployment algorithm is given in Algorithm \ref{algorithm:deployment}.
\end{remark}

\begin{algorithm}[h]
{\color{black}
\SetAlgoLined
Each robot $i$ sends its position $X_i$ and measurement $(X_i,y_i)$ to the center\;
The center estimates $\Hat{p}$ using $\{X_i\}_{i=1}^n$ and performs regression on $\mathcal{D}_0$ to obtain $\mathcal{V}_0$\;
\While{$\sup_x\mathcal{V}_k>\gamma$}{
    The center generates a reference model and sends $\{(p_{\text{r}}(x,t),v_{\text{r}}(x,t)\}_{t=0}^{T}$ to all robots\;
    \For{$t=t_c:\Delta t:t_c+T$}{
        Each robot computes its own velocity \eqref{eq:density feedback law using estimation} and sends its new position and measurement to the center\;
        The center estimates and sends $\Hat{p}$ to all robots\;
    }
    The center performs regression on $\mathcal{D}_k$ to obtain $\mathcal{V}_k$\;
}
\caption{Deployment algorithm}
\label{algorithm:deployment}
}
\end{algorithm}
}

% \begin{algorithm}
% \caption{Deployment algorithm}
% \begin{algorithmic}[1]
% % \Procedure{Roy}{$a,b$} 
% % \Comment{This is a test}
% \State Each robot $i$ sends its position $X_i$ and measurement $(X_i,y_i)$ to the center.
% \State The center estimates $\Hat{p}$ using $\{X_i\}_{i=1}^n$ and performs regression using $\mathcal{D}_0$ to obtain $\mathcal{V}_0$.
% \While{$\sup_x\mathcal{V}_k>\gamma$} 
%     \State The center generates a reference model and sends $\{(p_{\text{r}}(x,t),v_{\text{r}}(x,t)\}_{t=0}^{T}$ to all robots.
%     \For{$t=t_c:\Delta t:t_c+T$}
%     \State Each robot computes its own velocity command according to \eqref{eq:density feedback law using estimation} and sends its new position and new measurement to the center.
%     \State The center estimates $\Hat{p}$ and sends it to all robots.
%     \EndFor
%     \State The center performs regression using $\mathcal{D}_k$ to obtain $\mathcal{V}_k$.
% \EndWhile
% \end{algorithmic}
% \label{algorithm:deployment}
% \end{algorithm}

\subsection{Robustness of the control law}
Since we replace the original control law \eqref{eq:density feedback law} with \eqref{eq:density feedback law using estimation}, we should discuss its robustness issue with respect to density estimation errors.
A similar problem has been studied in our previous work \cite{zheng2020transporting}.
The control law \eqref{eq:density feedback law using estimation} is more complicated than the one in \cite{zheng2020transporting}.
% in which we studied a regulation problem, while this work studies a tracking control problem.
Nevertheless, the techniques used there apply to this work.
Before presenting the robustness results, we shall discuss the solution property of \eqref{eq:FP equation}.
We denote by $f_i$ the $i$-th component of a vector field $f$ and denote $\partial_iu:=\partial u/\partial x_i$ for a function $u(x,t)$.
The following lemma is proved in \cite{zheng2020transporting}.

\begin{lemma}[Well-posedness \cite{zheng2020transporting}]
\label{thm:well-posedness}
% (\textbf{Well-posedness} \cite{zheng2020transporting}). 
Assume 
\begin{equation}\label{eq:regularity condition1}
    v_i,\sigma,\partial_i\sigma\in L^\infty(\Omega\times(0,T)),\forall i\text{ and } p_0\in L^\infty(\Omega).
\end{equation}
Then there exists a unique weak solution $p$ for \eqref{eq:FP equation} with $p\in H^1(\Omega\times(0,T))$ and $\int_{\Omega}p(\cdot,t)dx=1$ for $\forall t\in(0,T]$. If we further assume
\begin{equation}\label{eq:regularity condition2}
    \partial_iv_i,\partial_i^2\sigma\in L^\infty(\Omega\times(0,T)),\forall i,
\end{equation} 
then $p_0\geq(\text{or}>)0$ implies $p\geq(\text{or}>)0$ for $t\in[0,T]$.
\end{lemma}

{\color{black}
Due to our choice of $K(x)$ in \eqref{eq:KDE}, $\hat{p}\in C^{\infty}(\Omega) \times C((0, T))$. 
Since $\Omega$ is bounded and $n$ is finite, we have $\inf_{x,t}\hat{p}(x, t)>0$ and $\sup_{x,t}\partial_{i}^{k}\hat{p}(x,t)<\infty$ for any $k\in\mathbb{N}$.
Therefore, if we additionally have $p_0>0$, $\sigma,p_{\text{r}},\in W^{2,\infty}(\Omega)\times L^\infty((0,T))$, $\alpha,v_{\text{r},i}\in W^{1,\infty}(\Omega)\times L^\infty((0,T))$, then the system \eqref{eq:FP equation} under \eqref{eq:density feedback law using estimation} satisfies the regularity conditions \eqref{eq:regularity condition1} and \eqref{eq:regularity condition2}.
These requirements are mild and can be easily satisfied in practice.

Now we discuss the robustness issue, which can arise in two situations.
First, any estimation algorithm contains estimation error due to finite samples.
Second, for \eqref{eq:density feedback law using estimation} to satisfy the regularity conditions, sometimes we have to correct $\Hat{p}$.
For example, we may want to add a small constant to $\Hat{p}$ and renormalize it when $\Hat{p}$ is close to 0, which introduces artificial estimation errors. 
We define $\epsilon:=\Hat{p}/p-1$. 
Then $\epsilon=0$ if and only if $\Hat{p}= p$, for which we treat $\epsilon$ as the estimation error.
Although seemingly unusual, this error model is more suitable for density estimates using KDE.
An equivalent way to define $\epsilon$ is that $\Hat{p}=p+\epsilon p$, where the additive noise is weighted by $p$.
It is known that under mild conditions, the estimation error of KDE is asymptotically Gaussian with a covariance proportional to the true density \cite{rao1986nonparametric}.
In other words, the estimation error is more uncertain when the true density is larger.
Hence, weighting the noise by $p$ makes $\epsilon$ a more uniform error model.
}
The robustness result is given below.

\begin{theorem}[ISS]
\label{thm:ISS}
Consider the PDE system \eqref{eq:FP equation} with control law \eqref{eq:density feedback law using estimation}. Assume the regularity conditions \eqref{eq:regularity condition1} and \eqref{eq:regularity condition2} are satisfied and $p_0>0$. 
Then the tracking error $\Phi$ is ISS in $L^2$ with respect to disturbances $\left\|\frac{\nabla\epsilon}{1+\epsilon}\right\|_{L^2}$ and $\left\|\frac{\epsilon}{1+\epsilon}\right\|_{L^2}$.
\end{theorem}

\begin{proof}
The proof is a modification of the proof of Theorem 3 in \cite{zheng2020transporting}.
First, Lemma \ref{thm:well-posedness} implies that $p(\cdot,t)$ is absolutely continuous.
Substituting \eqref{eq:density feedback law using estimation} into \eqref{eq:FP equation} and using $\hat{p}=p(1+\epsilon)$, we obtain
\begin{align*}
\begin{split}
    \partial_tp &= \nabla\cdot\left[p\frac{\alpha\nabla(\hat{p}-p_{\text{r}})-\nabla
    (\sigma\hat{p})-v_{\text{r}}p_{\text{r}}}{\hat{p}}\right]+\Delta(\sigma p)\\
    &=\nabla\cdot\left[\alpha\nabla(p-p_{\text{r}})+\frac{(\alpha-\sigma)p\nabla\epsilon+\alpha\epsilon\nabla p_{\text{r}}-v_{\text{r}}p_{\text{r}}}{1+\epsilon}\right].
\end{split}
\end{align*}
Consider a Lyapunov function $V(t)=\frac{1}{2}\int_\Omega(p-p_{\text{r}})^2dx$. 
By the Divergence theorem and the boundary condition, we have
\begin{align*}
    \Dot{V} &= \int_\Omega(p-p_{\text{r}})(\partial_tp-\partial_tp_{\text{r}})dx\\
    &= -\int_\Omega\nabla(p-p_{\text{r}})\cdot\Big[\alpha\nabla(p-p_{\text{r}})+v_{\text{r}}p_{\text{r}}\\
    &\qquad+\frac{(\alpha-\sigma)p\nabla\epsilon+\alpha\epsilon\nabla p_{\text{r}}-v_{\text{r}}p_{\text{r}}}{1+\epsilon}\Big]dx\\
    &= \int_{\Omega}-\alpha|\nabla\Phi|^{2}-\nabla\Phi\cdot\frac{(\alpha-\sigma)p\nabla\epsilon+(\alpha\nabla p_{\text{r}}+v_{\text{r}}p_{\text{r}})\epsilon}{1+\epsilon}dx
\end{align*}
Let $\alpha_{\min}:=\inf_{x,t}\alpha(x,t)>0$.
Since $p>0$, there exists a constant $\beta>0$ such that
\begin{align*}
    \Dot{V}\leq&-\alpha_{\min}\|\nabla\Phi\|_{L^2}^2+\beta\|\alpha-\sigma\|_{L^\infty}\|\nabla\Phi\|_{L^2}\left\|\frac{\nabla\epsilon}{1+\epsilon}\right\|_{L^2}\\
    &+\|\alpha\nabla p_{\text{r}}+v_{\text{r}}p_{\text{r}}\|_{L^\infty}\|\nabla\Phi\|_{L^2}\left\|\frac{\epsilon}{1+\epsilon}\right\|_{L^2}
\end{align*}
Fix a constant $\theta\in(0,1)$ to split the first term and apply the Poincar\'e inequality \cite{lieberman1996second} for the first two terms.
We obtain
\begin{equation*} 
    \begin{aligned}
    \Dot{V}
    &\leq -\frac{\alpha_{\min}(1-\theta)}{C^2} \|\Phi\|_{L^2}^2
    - \frac{\alpha_{\min}\theta}{C} \|\nabla\Phi\|_{L^2} \|\Phi\|_{L^2}\\
    &\quad +\beta\|\alpha-\sigma\|_{L^\infty}\|\nabla\Phi\|_{L^2}\left\|\frac{\nabla\epsilon}{1+\epsilon}\right\|_{L^2}\\
    &\quad +\|\alpha\nabla p_{\text{r}}+v_{\text{r}}p_{\text{r}}\|_{L^\infty}\|\nabla\Phi\|_{L^2}\left\|\frac{\epsilon}{1+\epsilon}\right\|_{L^2},
\end{aligned}
\end{equation*}
where $C>0$ is the constant from the Poincar\'e inequality.
Hence, we will have
\begin{equation*}
    \Dot{V}\leq -\frac{\alpha_{\min}(1-\theta)}{C^2} \|\Phi\|_{L^2}^2
\end{equation*}
if
\begin{align*}
    \|\Phi\|_{L^2}\geq&\frac{C\|\alpha\nabla p_{\text{r}}+v_{\text{r}}p_{\text{r}}\|_{L^\infty}\left\|\frac{\epsilon}{1+\epsilon}\right\|_{L^2}}{\alpha_{\min}\theta}\\
    &+\frac{C\beta\|\alpha-\sigma\|_{L^\infty}\left\|\frac{\nabla\epsilon}{1+\epsilon}\right\|_{L^2}}{\alpha_{\min}\theta}.
\end{align*}
where the right-hand side is a class $\mathcal{K}$ function of $\left\|\frac{\nabla\epsilon}{1+\epsilon}\right\|_{L^2}$ and $\left\|\frac{\epsilon}{1+\epsilon}\right\|_{L^2}$.
According to Lemma \ref{lmm:ISS Lyapunov functional}, we obtain the ISS property.
\end{proof}

This theorem means that with the mean-field feedback control law \eqref{eq:density feedback law using estimation}, the density tracking error $\|\Phi\|_{L^2(\Omega)}$ remains bounded in the presence of estimation error $\epsilon$ and converges asymptotically if $\epsilon=0$.

\section{Simulation studies}\label{section:simulation}
{\color{black}
In this section, we simulate a group of 100 robots to take measurements and reconstruct a sinc function $f(x)=2+\frac{\sin(2\|x\|)}{\|x\|}$ over $\Omega=[0,20]^2$ (see Fig. \ref{fig:sinc and prediction error}).
The robots' initial positions are drawn from a uniform distribution over $[0,7]^2$.
We follow Algorithm \ref{algorithm:deployment} to conduct the simulation, which contains two loops.
In the inner loop, the robots take measurements with a sampling period $\Delta t=0.875s$ and noise $\epsilon\sim\mathcal{N}(0,0.04)$.
In the outer loop, we update the reference model with a period $T=7s$, i.e. after receiving every 8 new data sets.
}
To generate the reference model, we follow the procedure in Section \ref{section:generate reference model} to obtain the target density $p_f$ and formulate the optimization problem, in which we set
\begin{align*}
    \phi(p(t_f),t_f)&=w_f[D(p(t_f)\|p_f)+D(p_f\|p(t_f))]\\
    L(p,v,t)&=w_p[D(p\|p_f)+D(p_f\|p)]+w_v\sum_{i=1}^{d}v_i^2
\end{align*}
where $w_f$, $w_p$ and $w_v$ are weights, and $D(p,q)=p\log\frac{p}{q}$.
{\color{black}
The optimality conditions \eqref{eq:optimality conditions} are solved using the generalized reduced gradient method given in \cite{rudd2013generalized}.
It discretizes the time axis into $N$ intervals and parameterizes every element $v_l\in\mathbb{R}^m,l=1,\dots,N$ using the Fourier sine series:
\[
v_l(x)=\sum_{i=1}^{I}\sum_{j=1}^{J}\sin (i\pi x_1/b)\sin(i\pi x_1/c)a_{ijl}(t)
\]
which forces $v_l(x)$ to be zero on the boundary of a rectangular work space $[0,b]\times[0,c]$ and hence satisfies the boundary condition of \eqref{eq:reference model}.
The coefficients $a_{ijl}\in\mathbb{R}^m$ are to be determined.
Here $N=12$, $m=2$ and $I=J=8$.
}
With initialized $a_{ijl}$, approximations of $p$ and $\lambda$ are obtained by numerically solving the state and co-state equations in \eqref{eq:optimality conditions}. 
Holding these two approximations fixed, the coefficients $a_{ijl}$ are updated by a gradient-based algorithm that minimizes the augmented cost.
This process is iterated until it converges to a local minimum that satisfies the optimization condition in \eqref{eq:optimality conditions}. 
The obtained coefficients $a_{ijk}$ are used to compute $v_{\text{r}}$ and generate $p_{\text{r}}$ according to \eqref{eq:reference model}.

\begin{figure}
\setlength{\abovecaptionskip}{0.0cm}
\setlength{\belowcaptionskip}{-0.3cm}
    \centering
    \begin{subfigure}[b]{0.23\textwidth}
        \centering
        \includegraphics[width=\textwidth]{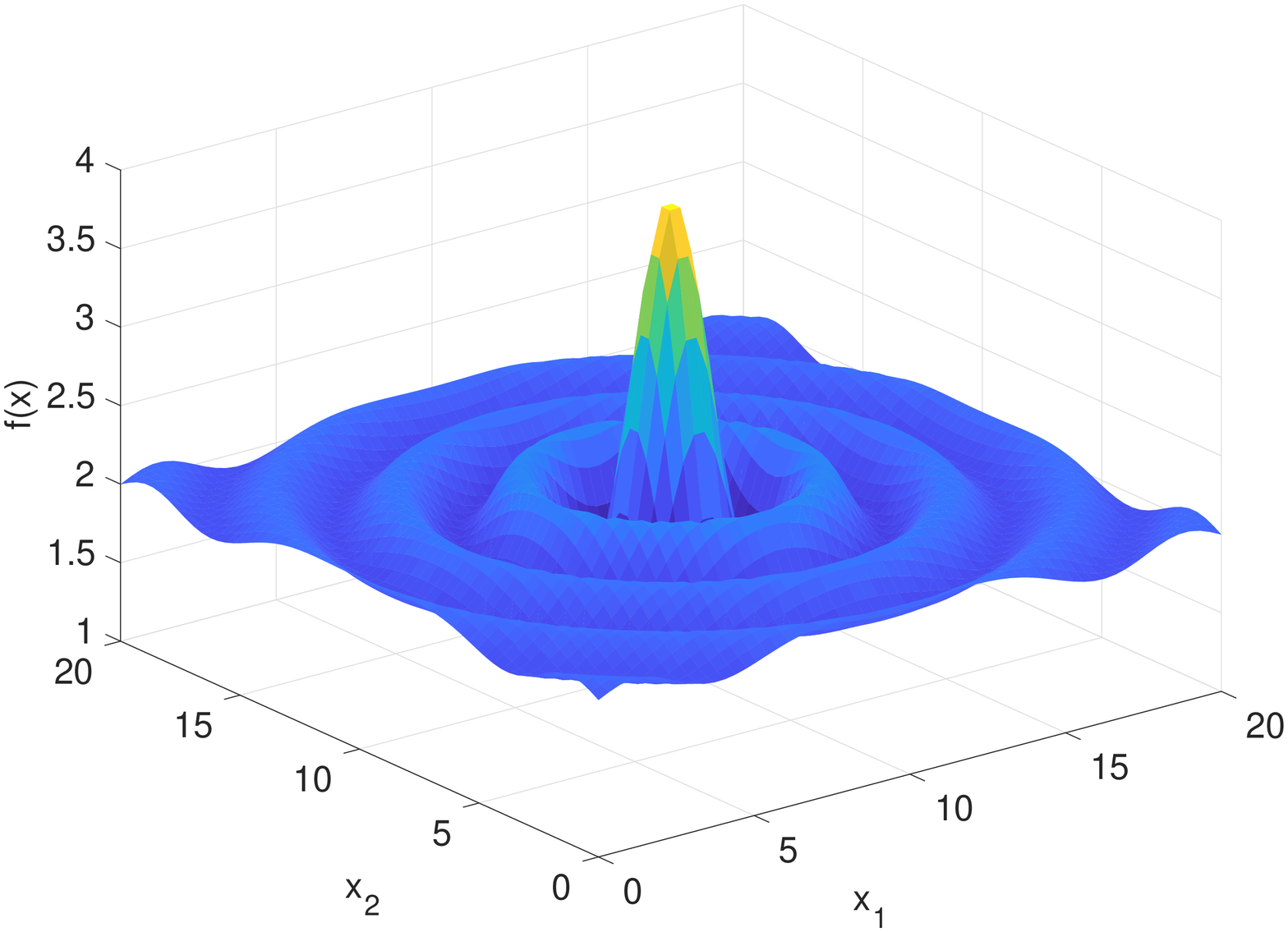}
    \end{subfigure}
    \begin{subfigure}[b]{0.23\textwidth}
        \centering
        \includegraphics[width=\textwidth]{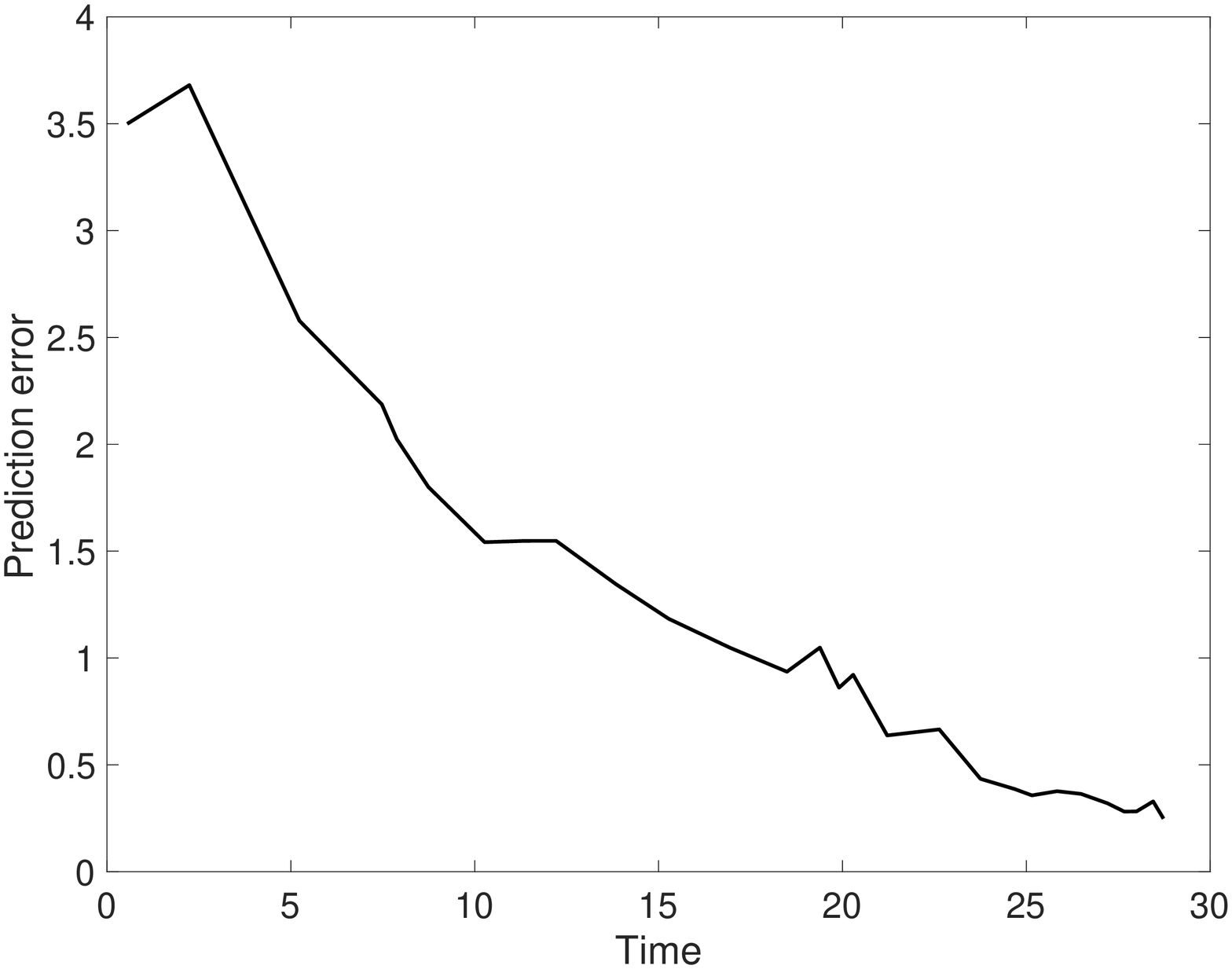}
    \end{subfigure}
    \caption{The unknown function (left); prediction error (right).}
    \label{fig:sinc and prediction error}
\end{figure}

\begin{figure*}[t]
\setlength{\abovecaptionskip}{0.0cm}
\setlength{\belowcaptionskip}{-0.5cm}
    \centering
    \begin{subfigure}[b]{0.23\textwidth}
        \centering
        \includegraphics[width=\textwidth]{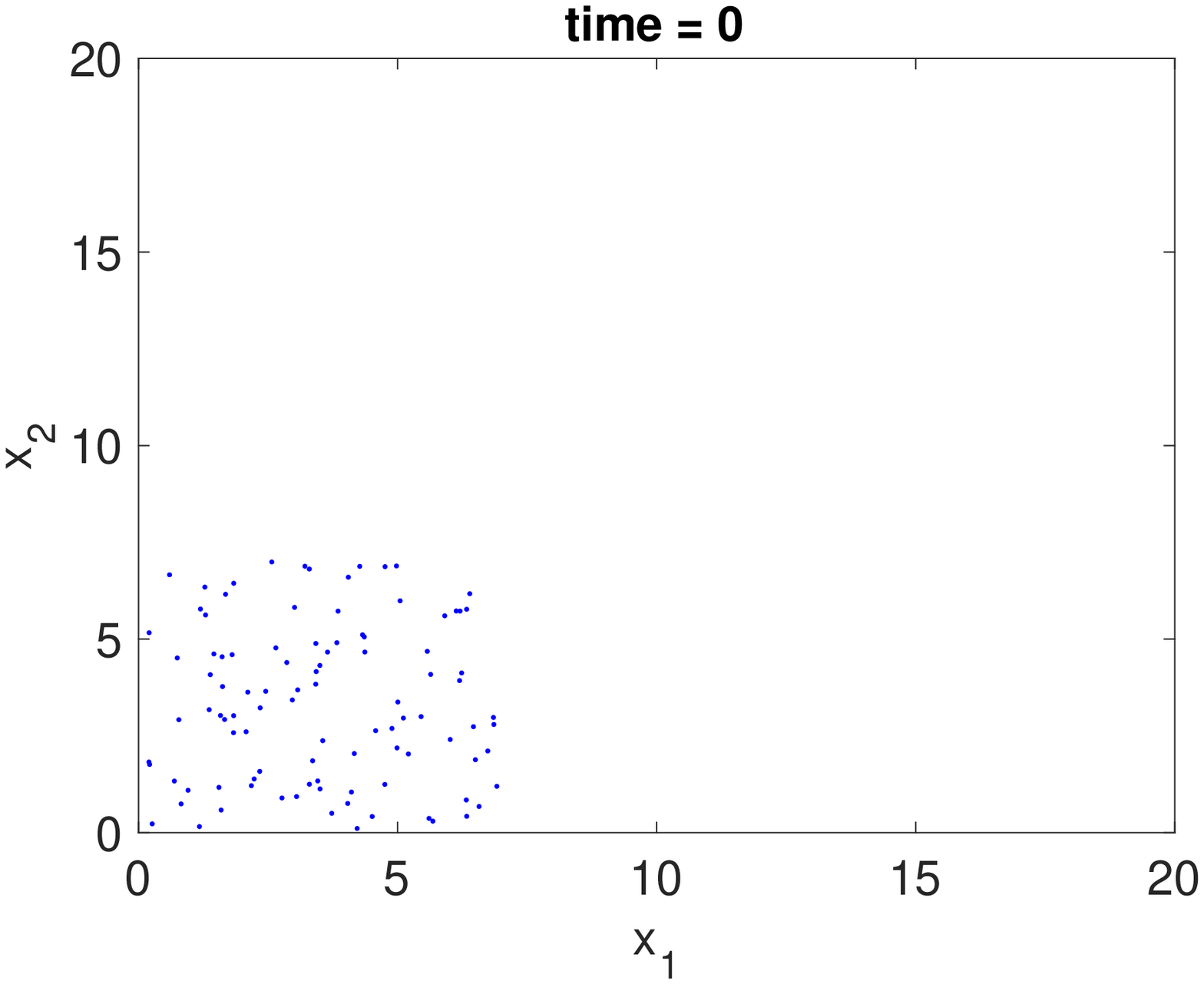}
    \end{subfigure}
    \begin{subfigure}[b]{0.23\textwidth}
        \centering
        \includegraphics[width=\textwidth]{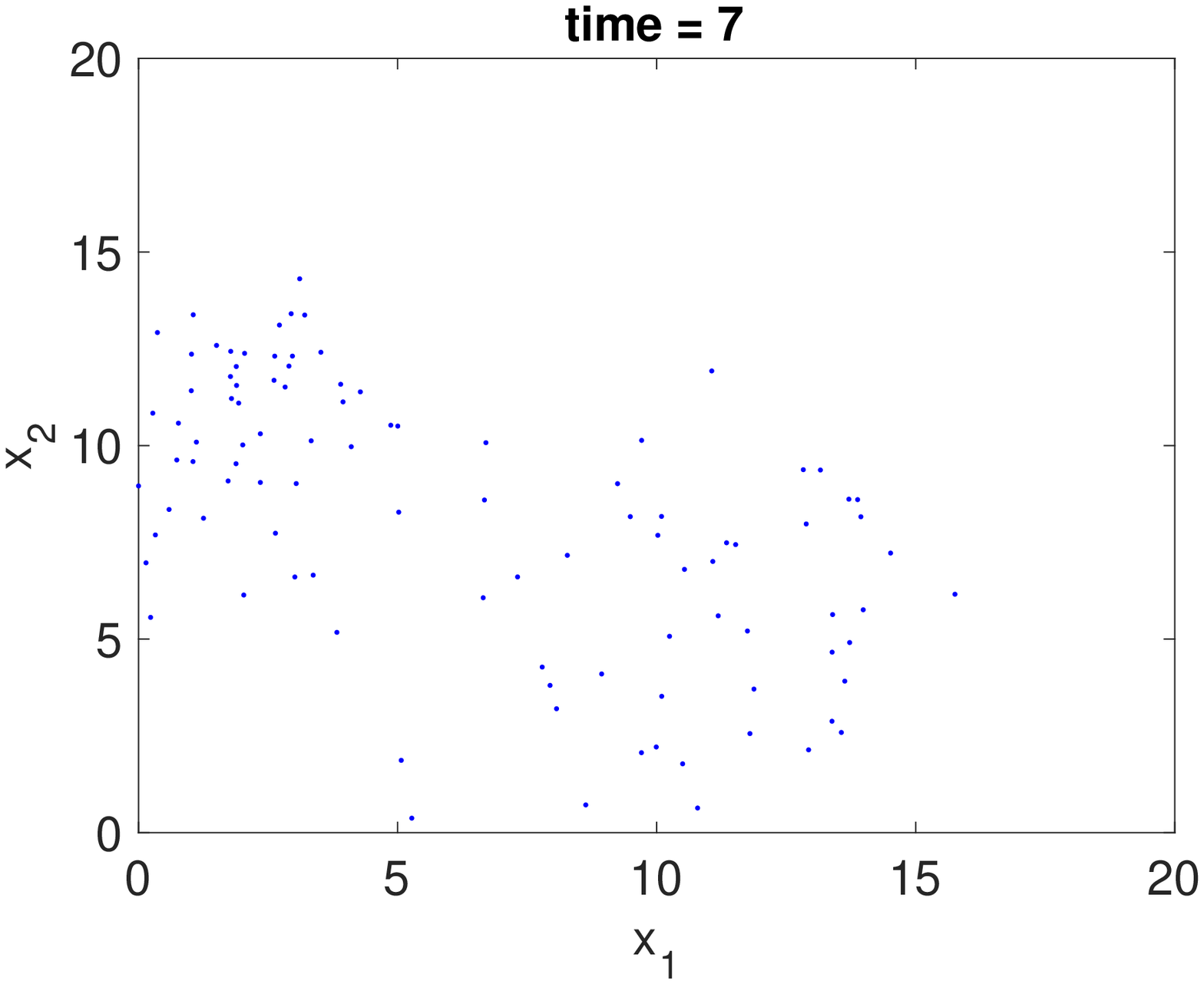}
    \end{subfigure}
    \begin{subfigure}[b]{0.23\textwidth}
        \centering
        \includegraphics[width=\textwidth]{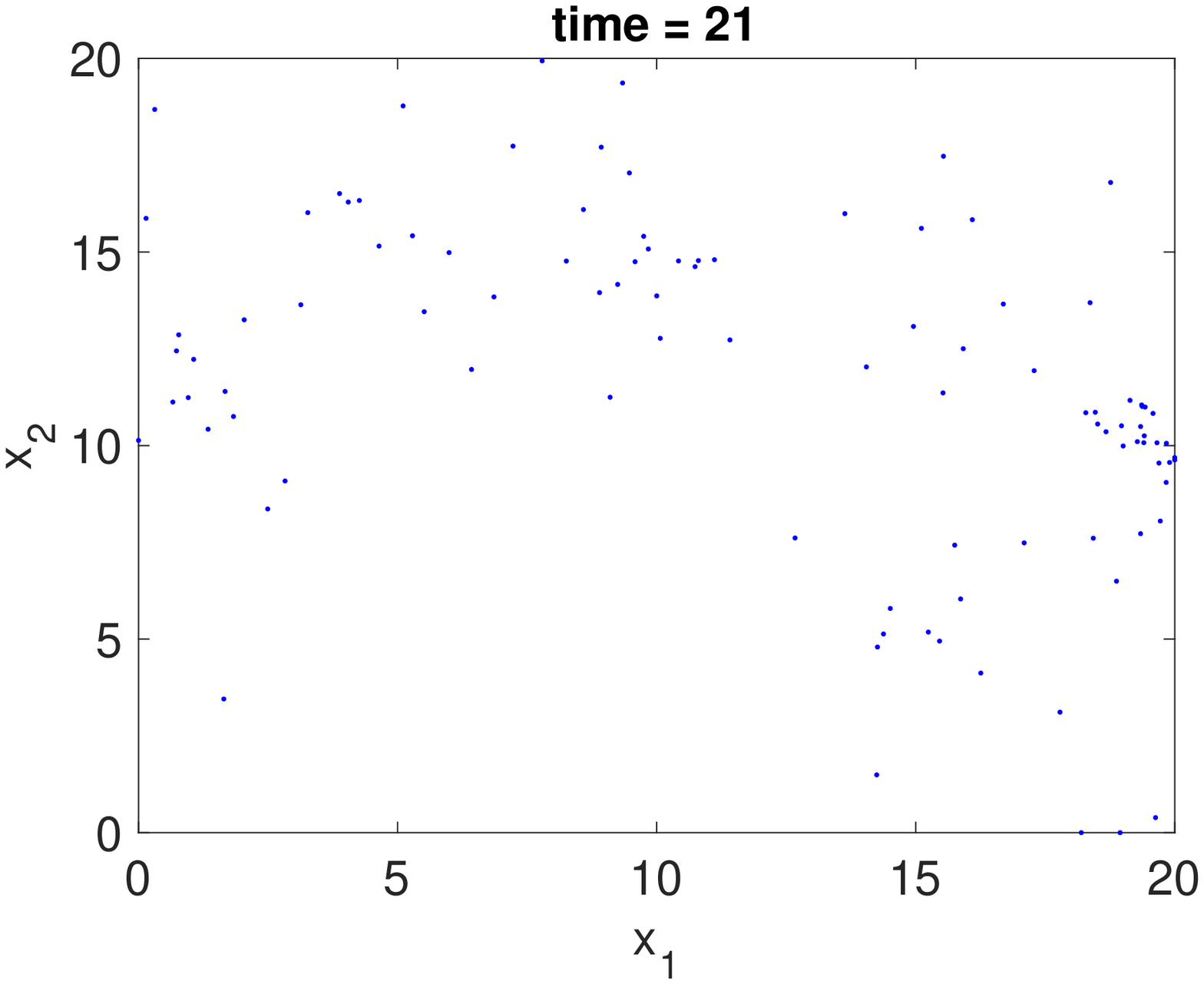}
    \end{subfigure}
    \begin{subfigure}[b]{0.23\textwidth}
        \centering
        \includegraphics[width=\textwidth]{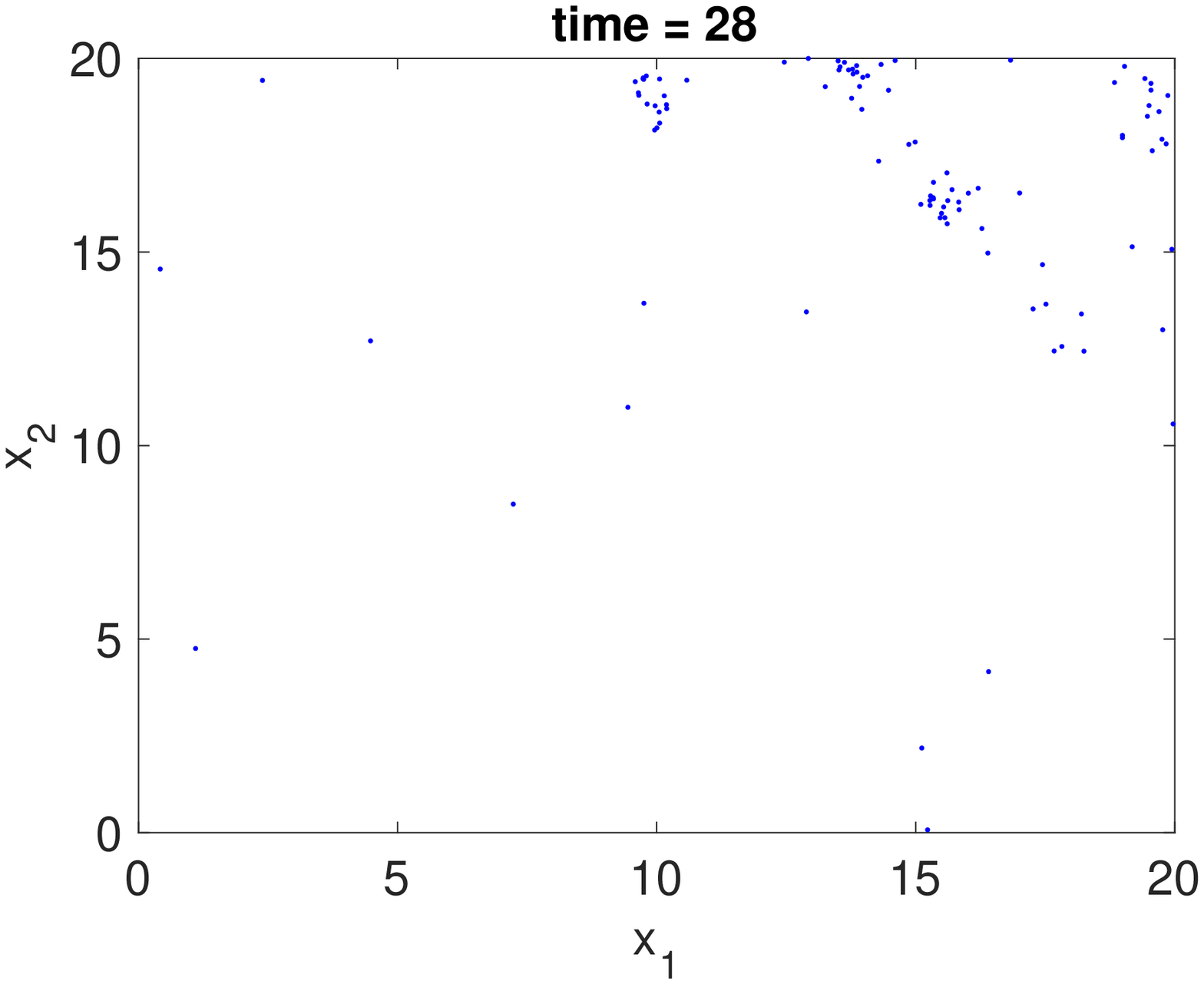}
    \end{subfigure}
    
    \begin{subfigure}[b]{0.23\textwidth}
        \centering
        \includegraphics[width=\textwidth]{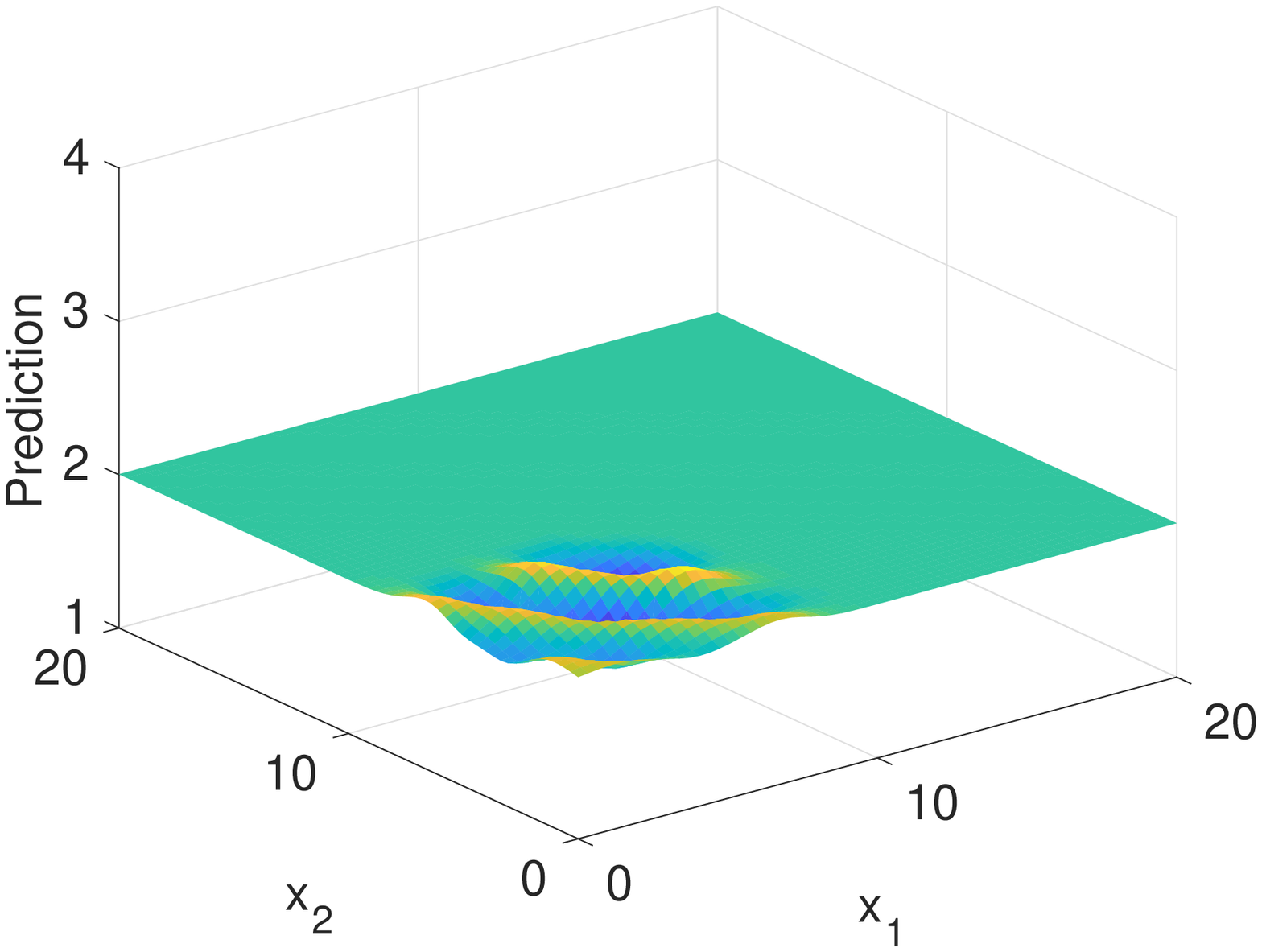}
    \end{subfigure}
    \begin{subfigure}[b]{0.23\textwidth}
        \centering
        \includegraphics[width=\textwidth]{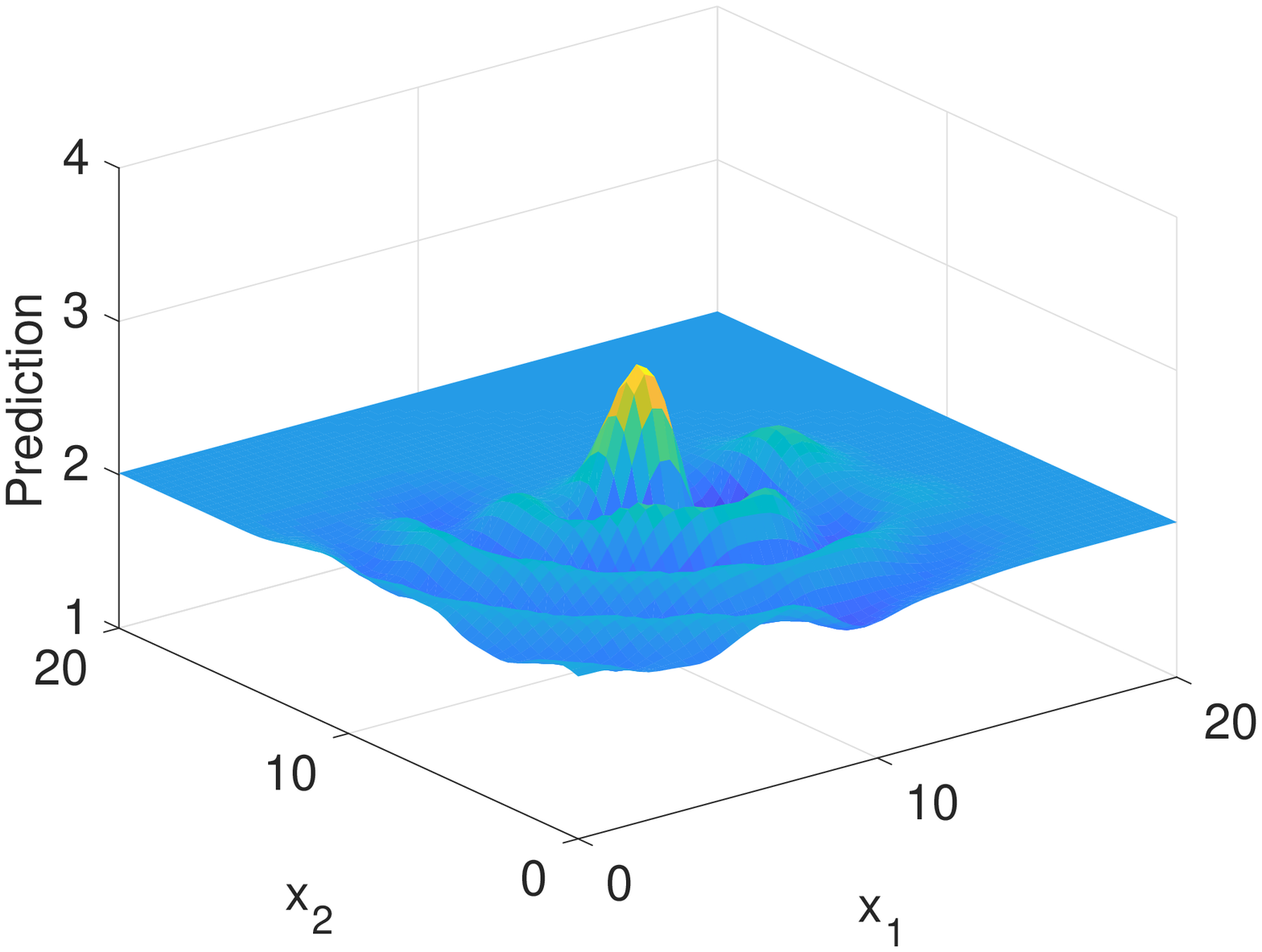}
    \end{subfigure}
    \begin{subfigure}[b]{0.23\textwidth}
        \centering
        \includegraphics[width=\textwidth]{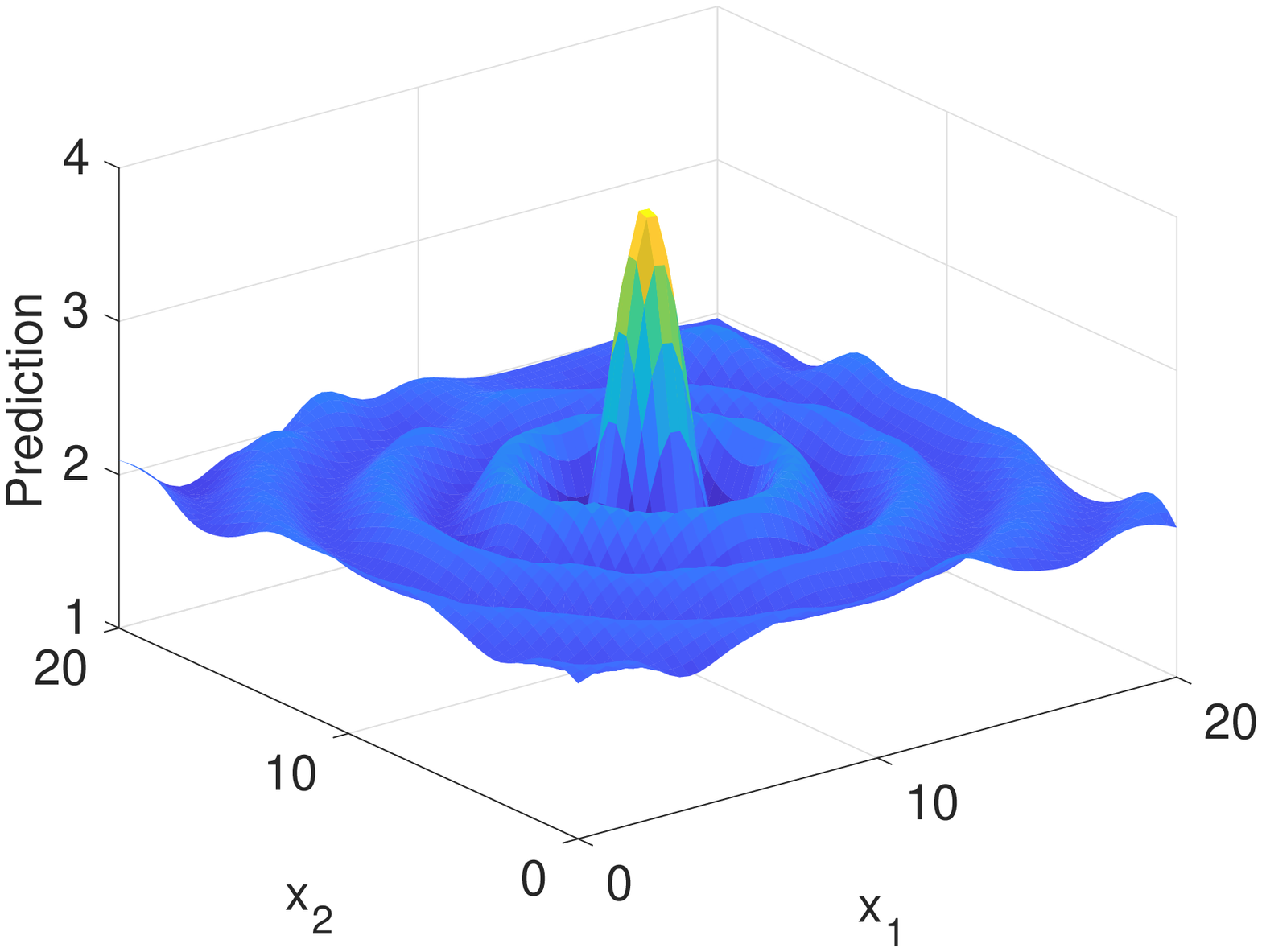}
    \end{subfigure}
    \begin{subfigure}[b]{0.23\textwidth}
        \centering
        \includegraphics[width=\textwidth]{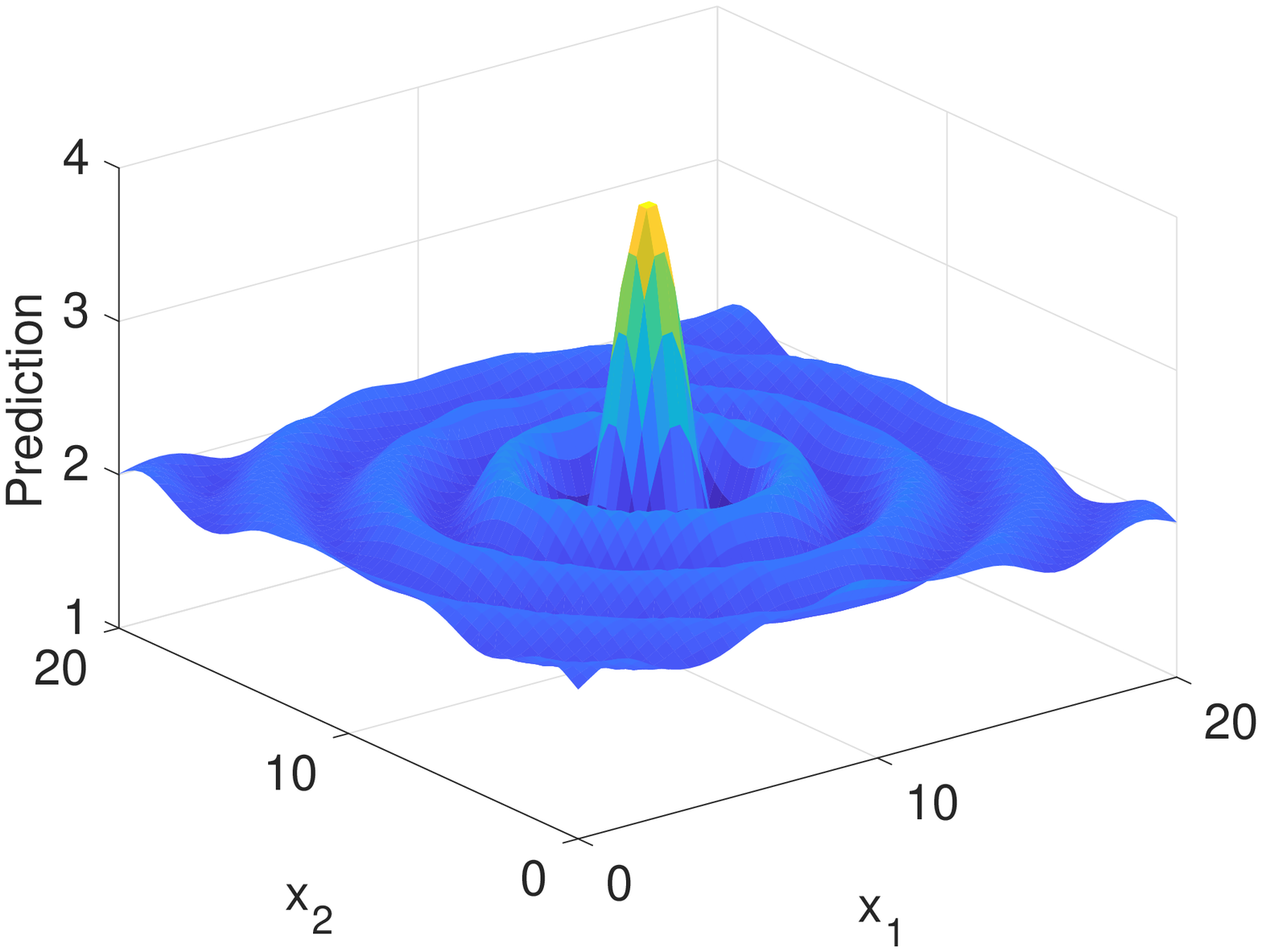}
    \end{subfigure}
    
    \begin{subfigure}[b]{0.23\textwidth}
        \centering
        \includegraphics[width=\textwidth]{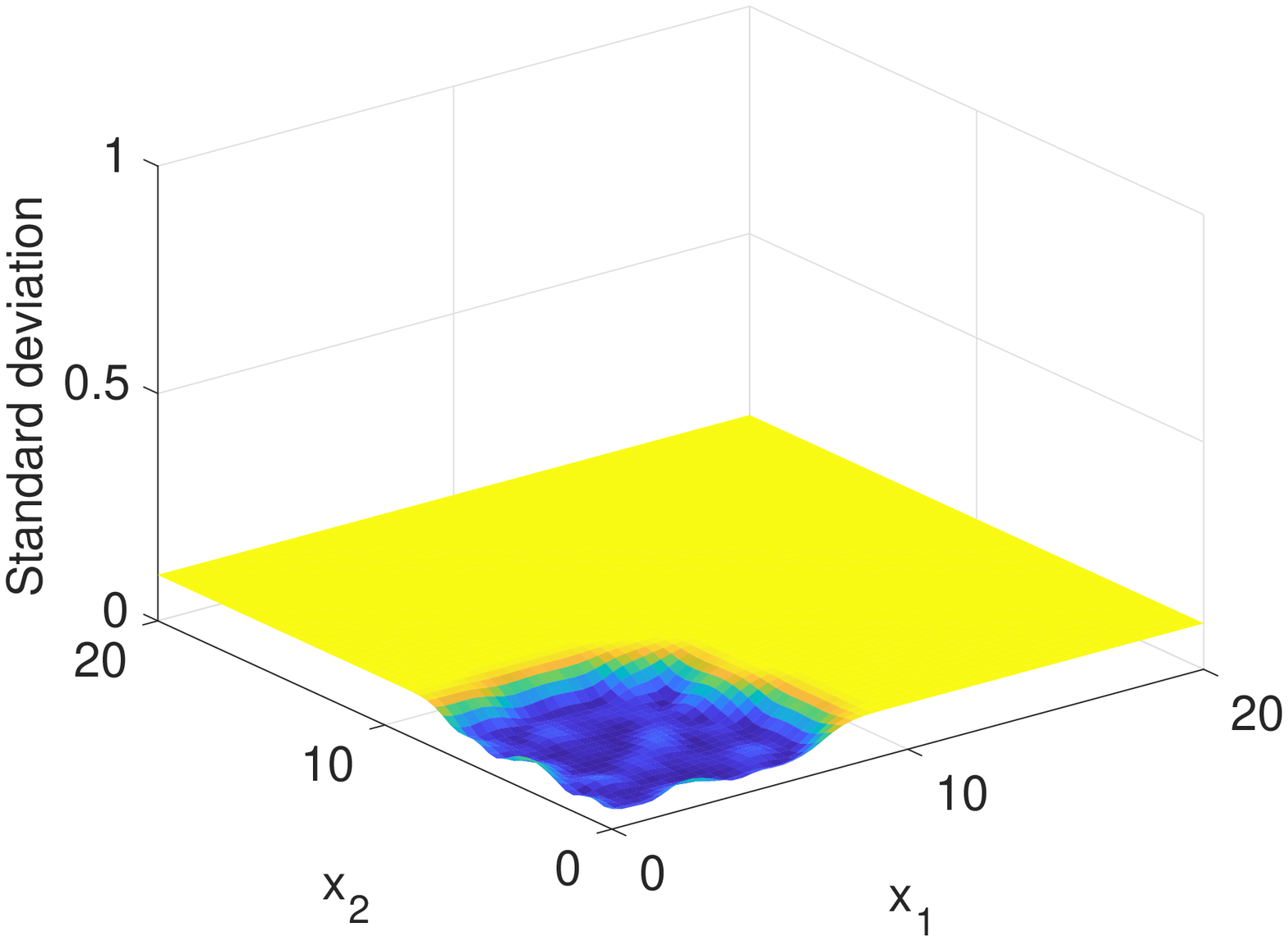}
    \end{subfigure}
    \begin{subfigure}[b]{0.23\textwidth}
        \centering
        \includegraphics[width=\textwidth]{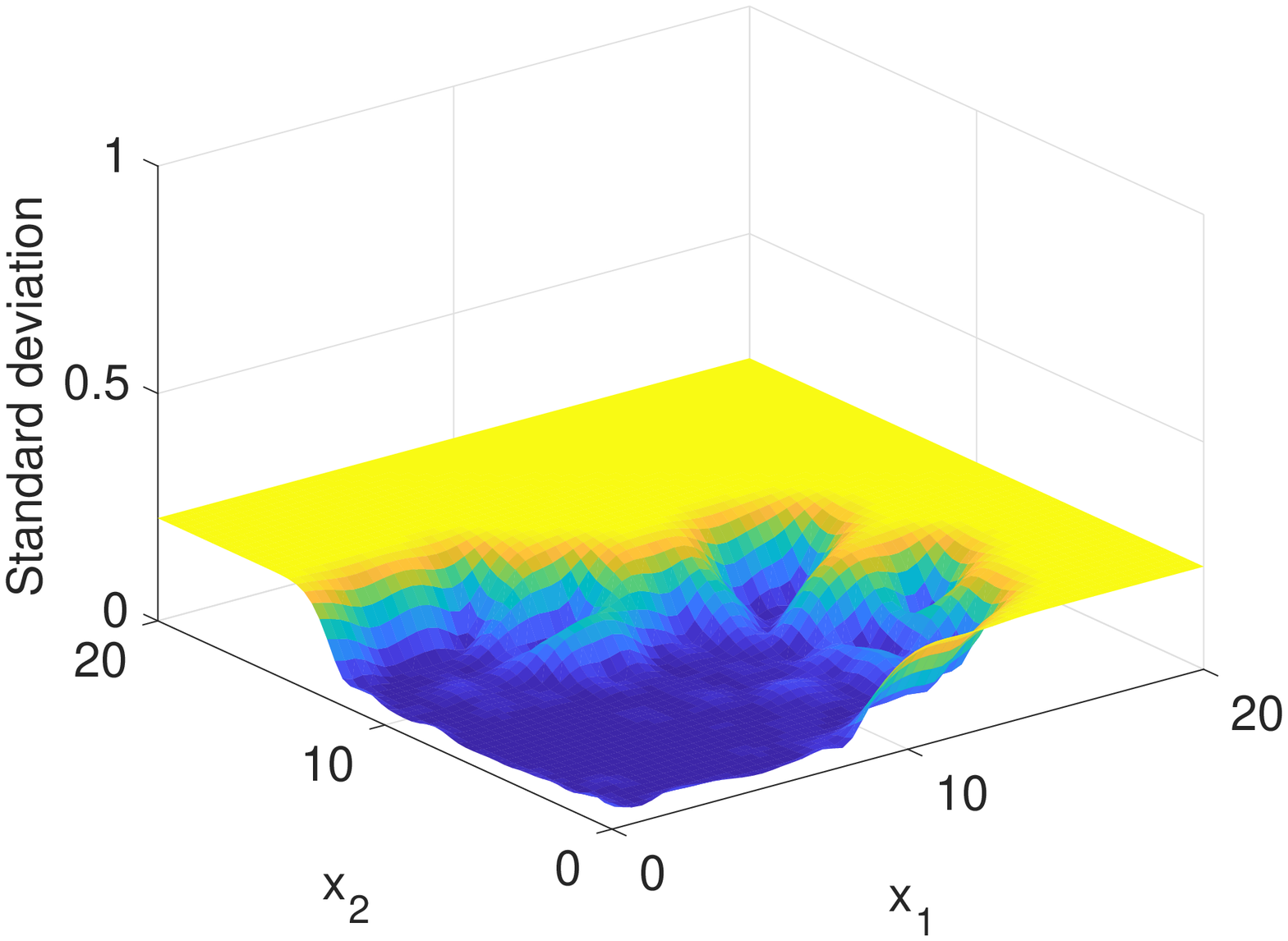}
    \end{subfigure}
    \begin{subfigure}[b]{0.23\textwidth}
        \centering
        \includegraphics[width=\textwidth]{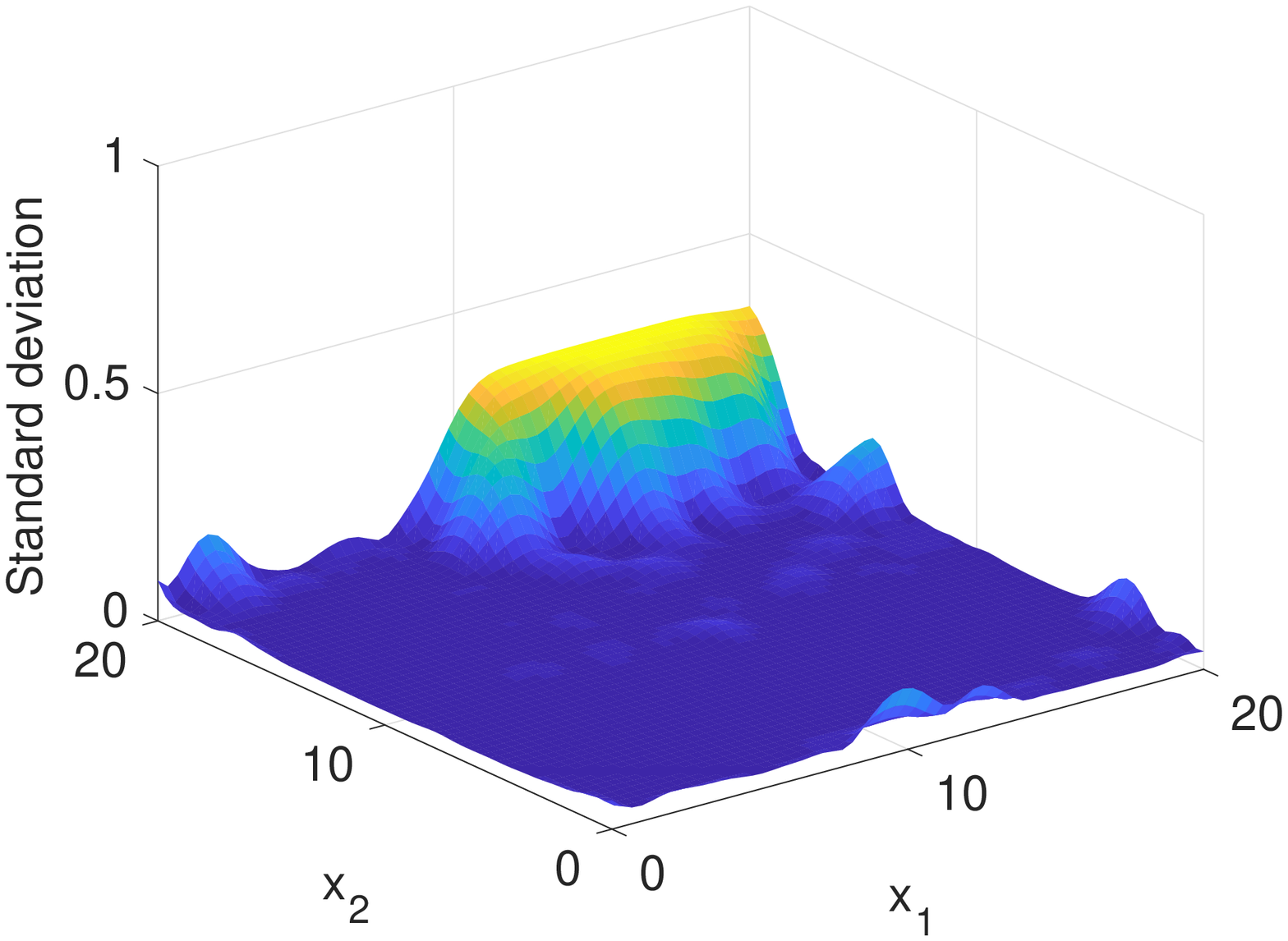}
    \end{subfigure}
    \begin{subfigure}[b]{0.23\textwidth}
        \centering
        \includegraphics[width=\textwidth]{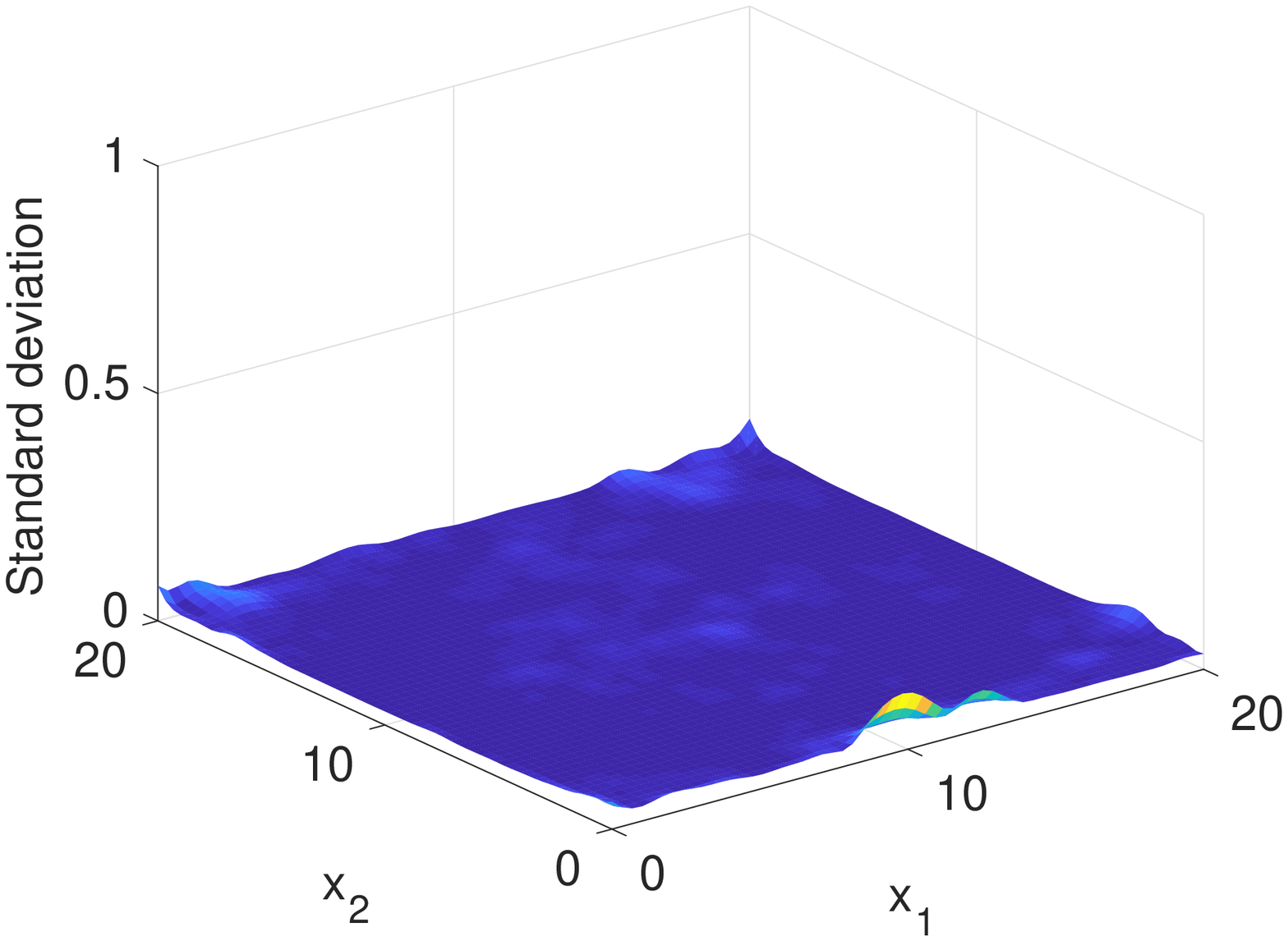}
    \end{subfigure}
    
    \begin{subfigure}[b]{0.23\textwidth}
        \centering
        \includegraphics[width=\textwidth]{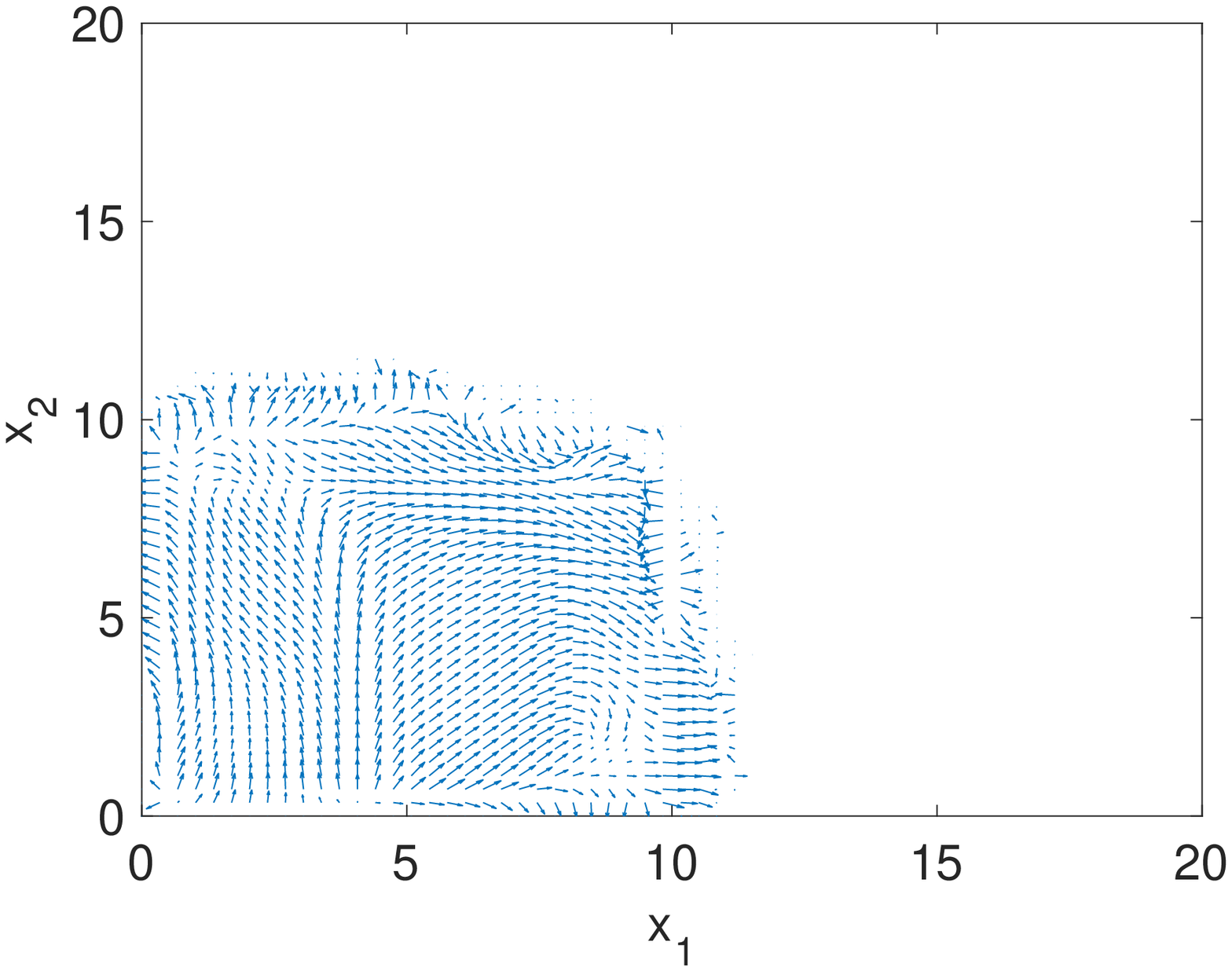}
    \end{subfigure}
    \begin{subfigure}[b]{0.23\textwidth}
        \centering
        \includegraphics[width=\textwidth]{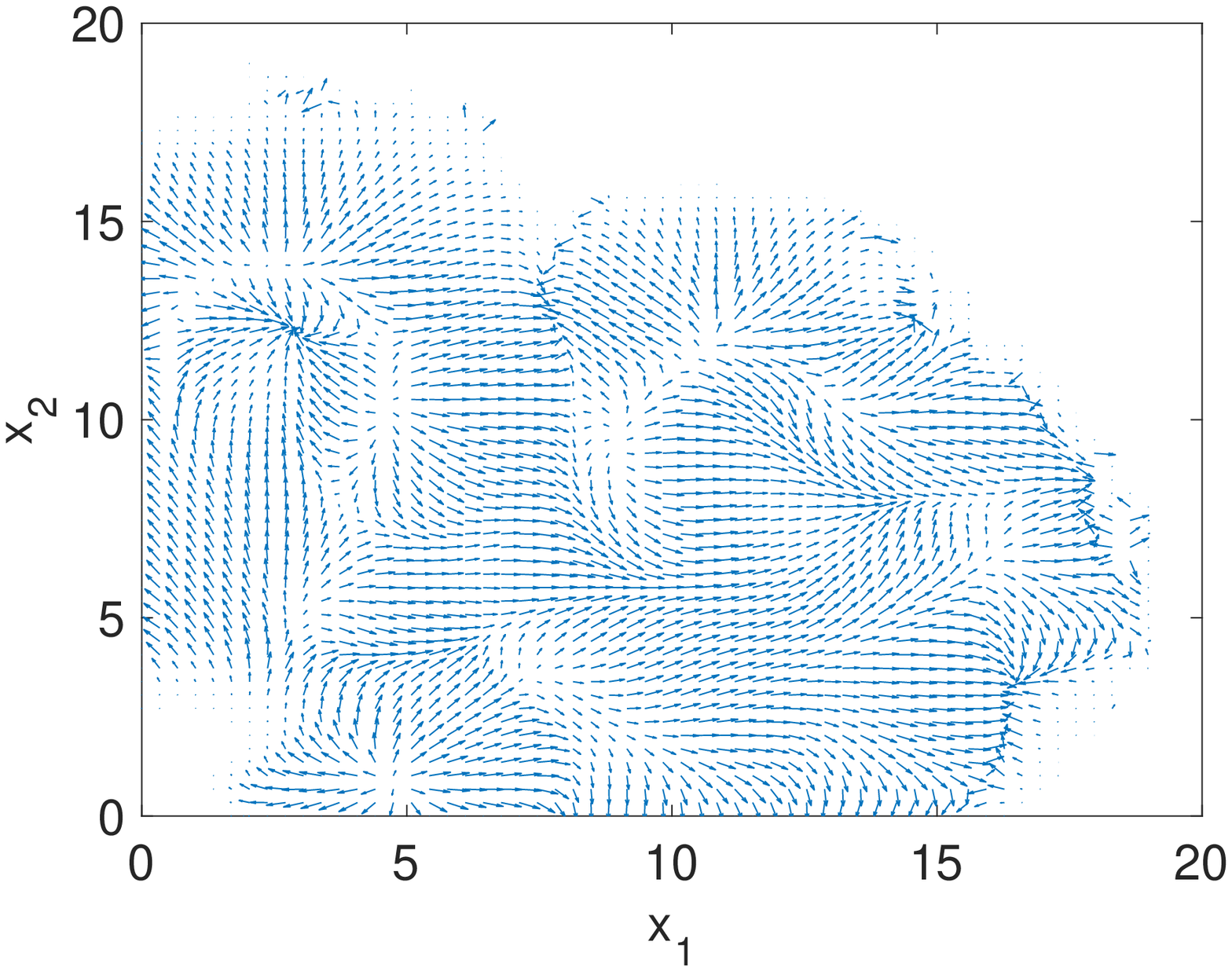}
    \end{subfigure}
    \begin{subfigure}[b]{0.23\textwidth}
        \centering
        \includegraphics[width=\textwidth]{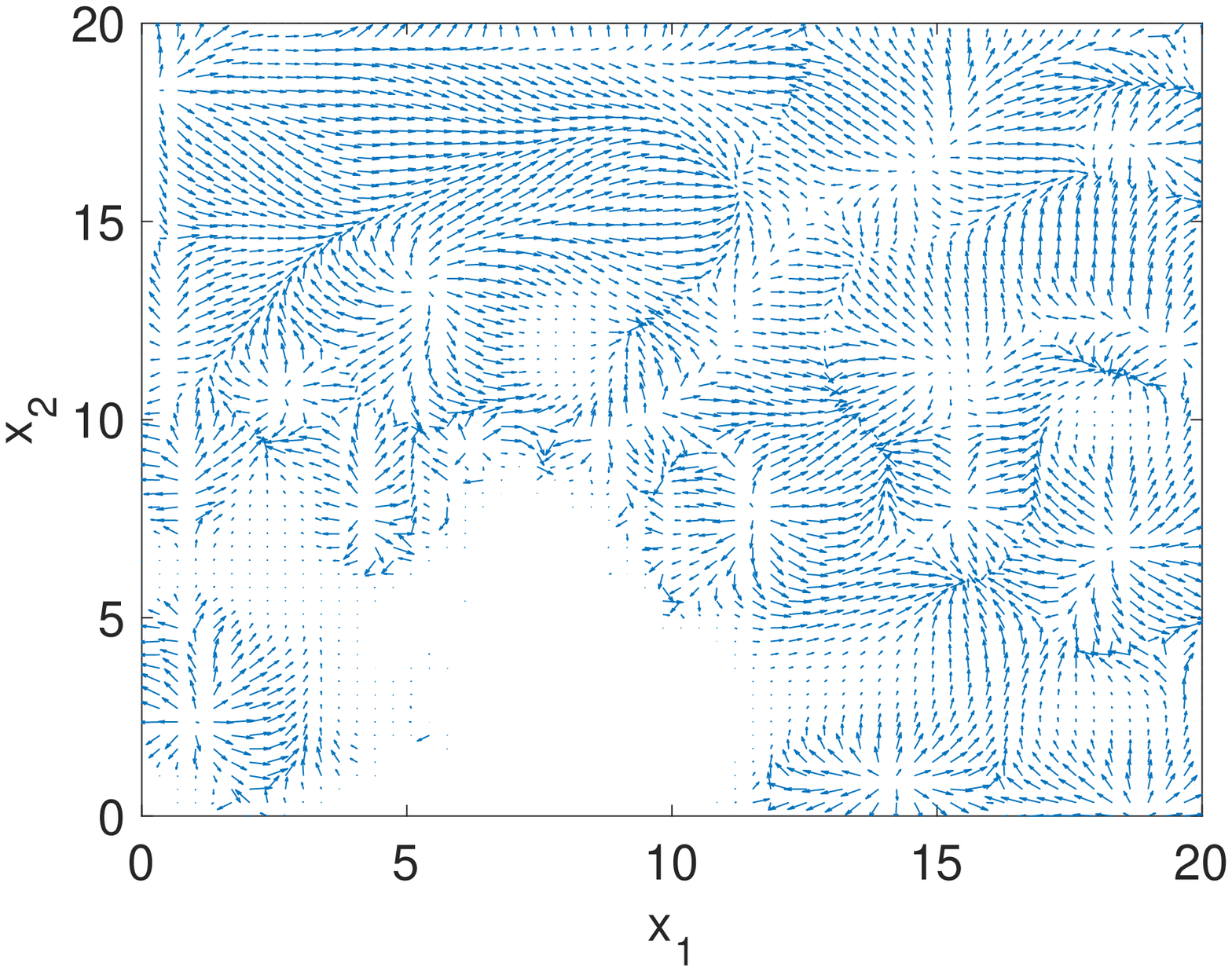}
    \end{subfigure}
    \begin{subfigure}[b]{0.23\textwidth}
        \centering
        \includegraphics[width=\textwidth]{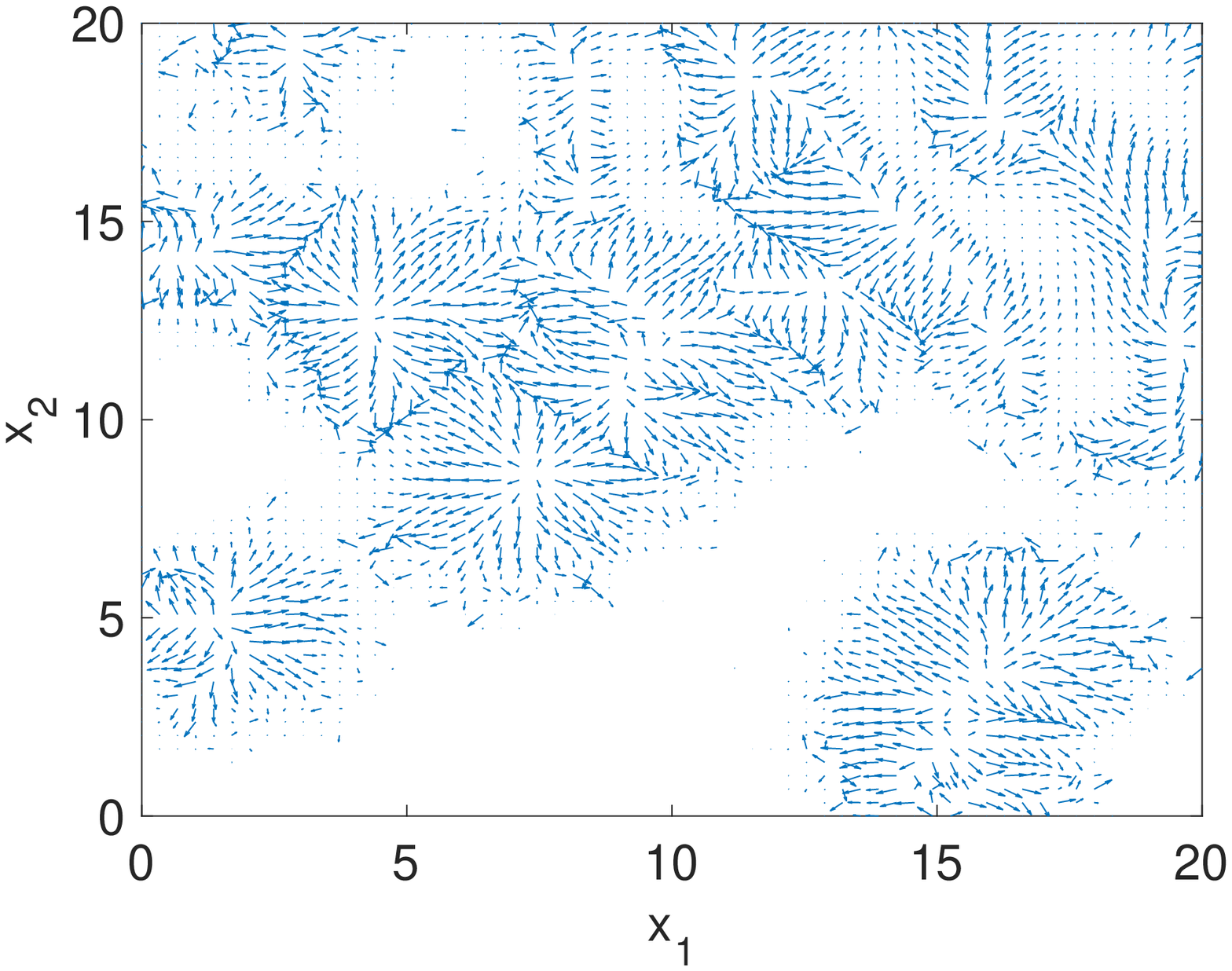}
    \end{subfigure}
    \caption{Movements of the robots (1st row), prediction $\Bar{f}(x)$ using GP regression (2nd row), predictive standard deviation $\mathcal{V}(x)$ of GP regression (3rd row) and generated velocity fields $v$ (4th row).}
    \label{fig:GP and density tracking}
\end{figure*}

The agents' trajectories are simulated using \eqref{eq:Langevin equation}.
Their velocity commands are computed using \eqref{eq:density feedback law using estimation}, where we use \eqref{eq:KDE} to estimate $\Hat{p}$.
{\color{black}
The predictive variance drops below the pre-specified $\gamma=0.1$ threshold after $4T$.
Simulation results are given in Fig. \ref{fig:GP and density tracking}, where the four columns represent four time steps $t=0,T,3T$ and $4T$.
The prediction error $\|f-\Bar{f}\|_{L^2}$ is given in Fig. \ref{fig:sinc and prediction error}.
It is seen that the robots keep moving to areas with larger predictive variance to take more measurements and eventually recover $f(x)$.
This verifies that the proposed algorithm is able to robustly guide the robots based on real-time performance.
}

\section{Conclusion and discussion}\label{section:conclusion}
This work presented a candidate framework for integrating machine learning techniques into the mean-field PDE-based approach for robotic swarms to achieve automatic deployment in real-time. 
We used BR models as an example to illustrate how to generate reference models based on the real-time prediction quality by formulating an optimization problem, and then designed density tracking laws such that the robots' density evolves according to the reference models.
The presented algorithms were applicable for field estimation tasks, while the idea of using machine learning for generating reference models and using density tracking laws for controlling the robots applies to a much wider range of applications of robotic swarms.
The proposed framework was essentially centralized.
Nevertheless, the control strategy was scalable to swarm sizes because each robot only needs to compute its own torque input to follow the velocity field.
Our future work will focus on decentralizing the mean-field feedback control algorithm by integrating the distributed density estimation algorithms that we recently proposed \cite{zheng2020pde, zheng2021distributed}.

%\addtolength{\textheight}{-12cm}   % This command serves to balance the column lengths
                                  % on the last page of the document manually. It shortens
                                  % the textheight of the last page by a suitable amount.
                                  % This command does not take effect until the next page
                                  % so it should come on the page before the last. Make
                                  % sure that you do not shorten the textheight too much.

%%%%%%%%%%%%%%%%%%%%%%%%%%%%%%%%%%%%%%%%%%%%%%%%%%%%%%%%%%%%%%%%%%%%%%%%%%%%%%%%

%%%%%%%%%%%%%%%%%%%%%%%%%%%%%%%%%%%%%%%%%%%%%%%%%%%%%%%%%%%%%%%%%%%%%%%%%%%%%%%%

%%%%%%%%%%%%%%%%%%%%%%%%%%%%%%%%%%%%%%%%%%%%%%%%%%%%%%%%%%%%%%%%%%%%%%%%%%%%%%%%
% \section*{APPENDIX}

% \section*{ACKNOWLEDGMENT}

\bibliographystyle{IEEEtran}
\bibliography{References}

\end{document}